\documentclass[11pt,oneside]{article}
\usepackage{graphicx}
\usepackage{fullpage}
\usepackage{amsmath}
\usepackage{amssymb}
\usepackage{amsthm}
\usepackage{tikz}
\usetikzlibrary{shapes}
\usetikzlibrary{patterns}
\usetikzlibrary{arrows,positioning,decorations.pathmorphing,trees}
\usetikzlibrary{arrows.meta}
\usepackage{pgfplots}
\usepackage{float}
\usepackage{hyperref}

\newtheorem{theorem}{Theorem}[section]

\newtheorem{corollary}[theorem]{Corollary}

\newtheorem{lemma}[theorem]{Lemma}

\newtheorem{definition}{Definition}

\newtheorem{example}{Example}
\newtheorem*{remark*}{Remark}
\newtheorem*{notation*}{Notation}
\newtheorem*{observation*}{Observation}
\newtheorem*{theorem*}{Theorem}
\newtheorem*{definition*}{Definition}
\newtheorem*{axiom*}{Axiom}
\newtheorem*{claim*}{Claim}
\newtheorem*{lemma*}{Lemma}

\title{Improved Two Sample Revenue Guarantees via\\ Mixed-Integer Linear Programming}
\author{Mete \c{S}eref Ahunbay \and Adrian Vetta}

\begin{document}

\maketitle

\begin{abstract}
We study the performance of the Empirical Revenue Maximizing (ERM) mechanism in a single-item, single-seller, single-buyer setting.
We assume the buyer's valuation is drawn from a regular distribution $F$ and that the seller has access to {\em two} independently 
drawn samples from $F$. By solving a family of mixed-integer linear programs (MILPs), the ERM mechanism is
proven to guarantee at least $.5914$ times the optimal revenue in expectation.
Using solutions to these MILPs, we also show that the worst-case efficiency of the ERM 
mechanism is at most $.61035$ times the optimal revenue. These guarantees improve upon
the best known lower and upper bounds of $.558$ and $.642$, respectively, of Daskalakis and Zampetakis~\cite{DZ20}.
\end{abstract}

\section{Introduction}

We study a primitive setting in revenue maximization: there is a single seller wishing to 
sell a single item to a single buyer, where the buyer's valuation for the item is drawn from a \emph{regular} 
distribution $F$ on $[0,\infty)$. Further, we incorporate the now widespread supposition that the 
valuation distribution $F$ is \emph{unknown to the seller}. Specifically, we present quantitative expected revenue
guarantees when the seller is allowed access to {\bf two} random, independently drawn sample valuations from $F$ 
before she selects a mechanism by which to sell the item.

When $F$ is known to the seller, Myerson~\cite{Myerson81} showed that the optimal mechanism the seller 
can implement is a \emph{posted price mechanism}. In a posted price mechanism, the seller 
chooses a price $p$ and the buyer decides to either buy the item or not. 
Of course, under the implementation of such a mechanism, the buyer would purchase the item if and only if his 
valuation for the item is greater than $p$. Given this, the seller simply picks a price $p$ which maximizes her expected 
revenue. Formally, denoting the probability she sells the item for price $p'$ as $1-F(p')$, the seller picks a
price $p \in \max_{p' \in \mathbb{R}_+} p' \cdot (1-F(p'))$.

But what about the case when $F$ is unknown to the seller? When the seller has \emph{sample access} to $F$, the natural approach 
is for the seller to assume the buyer's valuation distribution is given by the \emph{empirical distribution} $\hat{F}$ induced by
the set of samples; she may then simply implement the optimal mechanism of Myerson~\cite{Myerson81} using the empirical
distribution. This method, called the \emph{Empirical Revenue Maximising} (ERM) mechanism, provides surprisingly good performance
guarantees even in the case of a single sample. Specifically, Dhangwatnotai et al.~\cite{DRY10} showed that for the ERM mechanism just 
one sample suffices to give a $\frac12$-approximation to the optimal revenue. Huang et al.~\cite{HYR15} showed that this factor $\frac12$ bound 
is tight for any {\em deterministic} mechanism. In contrast, Fu et al.~\cite{FILS15} gave a {\em probabilistic} mechanism
obtaining at least $\frac12 + 5 \cdot 10^{-9}$ times the optimal revenue using a single sample. 

On the other hand, another line of work studies the performance of the ERM mechanism with respect to \emph{sample complexity}. 
This asks how many samples are necessary and/or sufficient to obtain a $(1-\epsilon)$-approximation of the optimal revenue, in expectation or with high probability. Dhangwatnotai et al.~\cite{DRY10} noted that even 
in our simple setting, the ERM mechanism does \emph{not} provide distribution independent polynomial sample complexity bounds; however, 
a \emph{guarded} variant of the ERM mechanism which ignores an $\epsilon$ fraction of the largest samples does produce a $(1-\epsilon)$-approximate 
reserve price with probability $(1-\delta)$  given $\Omega\left(\epsilon^{-3} \cdot \ln \left( \frac{1}{\epsilon\delta}\right)  \right)$ 
samples. Later, Huang et al.~\cite{HYR15} showed that any pricing algorithm that obtains a $(1-\epsilon)$-approximation 
of the optimal revenue requires $\Omega(\epsilon^{-3})$ samples, implying the factor $\epsilon^{-3}$ 
in the sample complexity result of~\cite{DRY10} is tight. For more on the sample complexity of the ERM
mechanism and its variants, see~\cite{ABGMMY17,CR14,DHP16,GN17,GW18,RS16,GHZ19}.

Motivated by the gap in our knowledge on sample complexity between the cases of a large number of samples and a single sample, 
Babaioff et al.~\cite{BGMM18} asked for revenue 
guarantees (in expectation) for a fixed number of samples $\ge 2$. Through a very rigorous case analysis they proved that, for \emph{two} samples, the 
ERM mechanism breaks the factor $\frac12$ barrier.
Specifically, it guarantees at least $.509$ times the optimal revenue in expectation. 
Significant improvements in revenue guarantees were then provided by Daskalakis and Zampetakis~\cite{DZ20}.
They showed that with two samples a \emph{rounded} version of the ERM mechanism obtains in 
expectation at least $.558$ times the optimal revenue. To achieve this they constructed a 
family of SDPs whose solutions provide lower bounds on the performance of the rounded ERM mechanism. Furthermore, through their 
primal solutions, they also showed that there exists a distribution of the buyer's valuation for which, with two 
samples, the ERM mechanism obtains in expectation at most $.642$ times the optimal revenue.

\subsection{Our Results}

In this paper, we study the ERM mechanism with two samples by building upon the optimization perspective of Daskalakis 
and Zampetakis~\cite{DZ20}. However, rather than an SDP-based framework we use an MILP-based framework to inspect 
the performance of the ERM mechanism in the single seller, single buyer, single item setting. 
This approach may seem impractical at first glance, given the hardness of mixed-integer linear programming
and the fact that, in general, there are no short certificates of optimality for solutions of MILPs. 
Our key technical contribution therefore is to present an MILP to bound the performance of the ERM mechanism
and which, despite the presence of $\geq 1000$ binary variables, can be approximately solved in a reasonable amount of time
with {\em provably} small error guarantees. This allows us to prove the ERM 
mechanism obtains at least $.5914$ times the optimal revenue. Furthermore, primal solutions to our MILPs show that 
there is a distribution $F$ for the buyer such that the ERM mechanism obtains at most $.61035$ times 
the optimal revenue. Consequently, we improve the lower bound on the revenue guarantee from
$.558$ to $.5914$ and the upper bound from $.642$ to $.61035$. Note also that we specifically 
analyze the ERM mechanism rather than the rounded version of ~\cite{DZ20}. Thus, in particular, we improve the
lower bound on guaranteed revenue under the ERM mechanism with two samples from the $.509$ of 
Babaioff et al.~\cite{BGMM18} to $.5914$.

\subsection{Overview of Paper}

In Section~\ref{sec:prelim} we present the problem of finding the worst case performance of the ERM mechanism 
with two samples. To motivate our MILP formulation, we prove the formulation of the problem in price space and 
quantile space are equivalent, and that for a given distribution $F$ of the buyer's valuation, the performance 
of the ERM mechanism can be calculated via an integral on $[0,1]^2$. In Section~\ref{sec:prog}, we present a 
class of MILPs which approximate or underestimate this integral, allowing us to compute provable upper and lower 
bounds on the performance of the ERM mechanism in our setting.  In Section~\ref{sec:res}, we show the resulting 
upper and lower bounds obtained by numerically solving these MILPs. Open problems and future directions 
are discussed in the conclusion.

\section{Preliminaries}\label{sec:prelim}

There are two agents: a seller and a buyer. The seller wishes to sell a single item to the buyer, 
whose valuation $v$ is drawn from a distribution $F$. To do so, the seller runs a posted price mechanism -- the 
seller commits to a price $p$, and the buyer can either take it or leave it. The buyer is utility maximizing, and his utility is 
quasilinear in payment. In particular, the buyer purchases the item if and only if $v \geq p$. Further, we make the standard
assumption that the distribution of the buyer's valuation, $F$, is \emph{regular}. A distribution $F$ on $\mathbb{R}_+$, 
given by its cumulative distribution function $F : \mathbb{R}_+ \rightarrow [0,1]$, is called regular if its revenue 
curve $R(q) = (1-q) \cdot F^{-1}(q)$ is concave on $(0,1)$. The objective of the seller is to maximize her revenue, but the 
distribution $F$ is unknown to her. Instead she must select the posted price based upon (two) independently 
drawn samples from~$F$.

\subsection{The Empirical Revenue Maximizing Mechanism}\label{sec:ERM}
Recall, by Myerson~\cite{Myerson81}, when $F$ is regular the optimal revenue mechanism for the seller is a posted price 
mechanism with price $p \in \arg \max_{p' \in \mathbb{R}_+} p' \cdot (1-F(p'))$. However, to implement such a mechanism, 
the seller would need to know $F$. But, without knowledge of $F$, how could she implement such a mechanism given only two 
independent samples, say $s \geq t$, drawn from $F$? We assume the seller does the obvious and 
implements the \emph{Empirical Revenue Maximizing} (ERM) mechanism. That is,
she simply posts a price which maximizes her expected revenue with respect to 
the \emph{empirical distribution} $\bar{F}$ she obtains via her two samples. Namely,
\begin{equation*}
\bar{F}(p) = \begin{cases}
0 & p < t \\
1/2 & t \leq p < s \\
1 & s \leq p
\end{cases}
\end{equation*}
Thus the seller sets price $p = t$ if $s < 2t$, and sets price $p = s$ if $s > 2t$. If $s = 2t$, since we are interested in 
worst case revenue, we may assume that the seller picks $p \in \{s,t\}$ which minimizes $p \cdot (1-F(p))$. Denote the 
expected revenue from posting price $p$ by $r(p) = p \cdot (1-F(p))$. Next, let the \emph{bisample expected revenue} $\psi_F(\cdot,\cdot)$ 
be defined as follows. When $s \geq t$, set
\begin{equation}
\psi_F(s,t) = r(s) \cdot \mathbb{I}(s > 2t) + r(t) \cdot \mathbb{I}(s < 2t) + \min\{r(s),r(t)\} \cdot \mathbb{I}(s = 2t),
\end{equation}
and, when $s < t$, set $\psi_F(s,t) = \psi_F(t,s)$. Then the seller's revenue for implementing the ERM mechanism is exactly:
\begin{equation}\label{eqn:priceSpace}
\bar{r}_F = \int_{(s,t) \in \mathbb{R}^2_+}\psi_F(s,t) \cdot dF(s) \times dF(t)
\end{equation}
In turn, the optimal revenue for distribution $F$ is given by $r_F = \max_{p \in \mathbb{R}_+} p \cdot (1-F(p))$. In this paper, 
we are interested in providing lower and upper bounds for the relative performance of the ERM mechanism and the 
optimal mechanism, $\alpha = \inf_{F | F \text{ is regular}} \bar{r}_F / r_F$.

\subsection{Revenue Curves and the Quantile Space}
Following Daskalakis and Zampetakis~\cite{DZ20}, we will deriving our bounds on $\alpha$ via a reduction to a set of optimization 
programs for which we compute solutions. Importantly, however, we make a very different choice of variables.
Specifically, we choose to work in the {\em quantile space} (i.e. with the revenue curve) rather than working in the {\em price space} 
(i.e. with the PDF/CDF of the distribution directly).
Towards this end, note first that if~$R$ is the revenue curve of the distribution~$F$ then $R(q) = (1-q) \cdot F^{-1}(q)$ for any $q \in [0,1]$. 
Therefore, for each $q \in [0,1]$, the revenue curve provides the \emph{price inverse} of $q$:
\begin{align*}
F^{-1}(q) = \begin{cases}
R(q)/(1-q) & q \in [0,1) \\
\lim_{q' \rightarrow 1^-} R(q')/(1-q') & q = 1
\end{cases}
\end{align*}

Via the price inverse, we may define the \emph{bisample revenue function} $\phi_R(\cdot,\cdot)$ on $[0,1]^2$. 
To do this, if $(x,y) \in [0,1]^2$ and $x \geq y$, set
\begin{align}
\phi_R(x,y) &\ =\  R(x) \cdot \mathbb{I}[F^{-1}(x) >  2 F^{-1}(y)] \ +\  R(y) \cdot \mathbb{I}[F^{-1}(x) <  2 F^{-1}(y)]  \label{def:phiR}\\ 
& \qquad\qquad +\  \min\{R(x),R(y)\} \cdot \mathbb{I}[F^{-1}(x) =  2 F^{-1}(y)] \nonumber
\end{align}
If instead $x < y$, we symmetrically extend the function by setting $\phi_R(x,y) = \phi_R(y,x)$.

We want to write $\bar{r}_F$ as a double integral on $[0,1]^2$. By (\ref{eqn:priceSpace}), this integral should have the form:
\begin{equation}
\bar{r}_R = \int_{(x,y) \in [0,1]^2} \phi_R(x,y) \cdot d(x,y) \label{eqn:qSpace}
\end{equation}

We remark  that there are several immediate advantages of such a choice of variables. First, in price space, the regularity 
constraints are highly non-linear; as shown by Daskalakis and Zampetakis~\cite{DZ20}, given a gridding of the real line, 
each regularity constraint is a degree three polynomial inequality. However, in quantile space, we can impose the 
regularity of the distribution with the following set of \emph{linear} constraints:
$$ \forall x,y,\lambda \in [0,1], \qquad R(\lambda x + (1-\lambda) y ) \geq \lambda R(x) + (1-\lambda) R(y)$$ 
Furthermore, working in quantile space, we only need to approximate an integral on $[0,1]^2$ with our optimization formulation, 
thereby avoiding the need to explicitly eliminate the tails of the distribution. Instead, the tiling of $[0,1]^2$ and normalisation of the 
optimal revenue naturally places an upper bound on $F^{-1}$ whenever $R(1) = 0$.

Unfortunately, working in the quantile space also has disadvantages. While we can avoid non-linearities in the regularity constraints, 
we will still need some quadratic constraints, and quadratic (or cubic) objective functions. In particular, both the feasible region 
and the objective function will be non-convex. Furthermore, we will work with an MILP formulation, using $\sim \binom{n}{2}$ binary 
variables for a grid of $n$ intervals to handle degree 2 or degree 3 polynomial terms. Handling such a large set of binary variables for large $n$ 
will require care in how we set up our optimization programs.

\subsection{The Validity of the Quantile Space Approach}
There is also a subtle technical complication we must address. Every regular distribution has a concave revenue curve. 
To get our approach of using revenue curves to work, we then want to write:
\begin{equation}
\alpha = \inf_{F | F \text{ is regular}} \bar{r}_F / r_F  = \inf_{R | R : [0,1] \rightarrow \mathbb{R}_+ \text{ is concave}} \bar{r}_R / r_R \label{eqn:equiv}
\end{equation}
where $r_R = \max_{q \in [0,1]} R(q)$. However, it is not the case that the integrals (\ref{eqn:priceSpace}) and (\ref{eqn:qSpace}) 
are equal for each regular distribution, given our definition of regularity. For example:
\begin{example} Consider the point mass at $1$:
\begin{equation*}
F(x) = \begin{cases}
0 & x < 1 \\
1 & x \geq 1
\end{cases}
\end{equation*}
Then $F^{-1}(q) = 1$ for every $q \in (0,1]$, so the revenue curve is given $R(q) = 1-q$. Here, the revenue curve is 
concave -- in fact, affine. However, $\bar{r}_F = 1$ while $\bar{r}_R = 2/3$. 
\end{example}

On the other hand, for any $\epsilon > 0$, the uniform distribution on $[1,1+\epsilon]$ has revenue curve $R_\epsilon(q) = (1-q) \cdot (1+q\epsilon)$, a 
concave function. Furthermore, for any $\epsilon > 0$, $\bar{r}_{F_\epsilon} = \bar{r}_{R_\epsilon}$, and in the 
limit $\epsilon \rightarrow 0$, $\bar{r}_{R_\epsilon} \downarrow 2/3$ as intended. This equivalence is caused by each 
distribution $F_\epsilon$ being continuous (i.e. having a probability density function). Consequently, to prove the equivalence of 
our formulation (\ref{eqn:equiv}), we must prove a convergence result:

\begin{theorem}\label{thm:convergence}
Suppose that $R : [0,1] \rightarrow \mathbb{R}_+$ is concave. Then there exists a sequence $(F_n)_{n \in \mathbb{N}}$ of 
continuous regular distributions such that the sequence of associated revenue curves $(R_n)_{n \in \mathbb{N}}$ converges uniformly to $R$ on $[0,1]$.
\end{theorem}

To prove Theorem~\ref{thm:convergence}, we need to prove several properties of concave, non-negative functions on $[0,1]$. The first result 
characterizes in a sense when $R$ {\em fails} to be the revenue curve of a continuous distribution:

\begin{lemma}\label{lem:increase}
Suppose that $R : [0,1] \rightarrow \mathbb{R}_+$ is concave (strictly concave), then $F^{-1}$ is non-decreasing (strictly increasing) 
on $[0,1)$. Moreover, $F^{-1}$ can only fail to strictly increase on a final segment: if for some $y < x$ in $[0,1)$ we 
have $F^{-1}(x) = F^{-1}(y)$, then for any $z \in [y,1), F^{-1}(z) = F^{-1}(y)$.
\end{lemma}
\begin{proof}
Let $x,y \in [0,1)$ such that $y < x$. Then $x = \lambda \cdot y + (1-\lambda) \cdot 1$ for $\lambda = \frac{1-x}{1-y} \in (0,1)$. Now, 
as $R$ is concave, $R(x) \geq \lambda \cdot R(y) + (1-\lambda) \cdot R(1)$. However, $R$ is non-negative, so $R(1) \geq 0$. In 
particular, $R(x) \geq \lambda \cdot R(y)$. Because $\lambda = \frac{1-x}{1-y}$, we then have $\frac{R(x)}{1-x} \geq \frac{R(y)}{1-y}$.

To show the second statement, suppose that $F^{-1}(x) = F^{-1}(y)$. The case when $z \in (y,x)$ holds since $F^{-1}(\cdot)$ is non-decreasing.
So suppose that $z \in (x,1)$. Then $x = \lambda \cdot y + (1 - \lambda) \cdot z$. If $F^{-1}(z) > F^{-1}(y)$ then $R(z)(1-y) > R(y)(1-z)$. 
However, by definition of $\lambda$, we have $(1-x) = \lambda(1-y) + (1-\lambda)(1-z)$. This implies that
\begin{align*}
R(y) (1-x) & = \lambda R(y) (1-y) + (1-\lambda) (1-z) R(y) \\
& < \lambda R(y) (1-y) + (1-\lambda) (1-y) R(z) \\
& = (1-y) (\lambda R(y) + (1-\lambda) R(z)) \\
& \leq (1-y) R(x).
\end{align*}
Here the first inequality holds by the assumption $F^{-1}(z) > F^{-1}(y)$. The second inequality holds since $R$ is concave. 
This implies that $F^{-1}(x) > F^{-1}(y)$, a contradiction. Therefore, if $z \in (x,1)$ then $F^{-1}(z) = F^{-1}(y)$. 
\end{proof}

This is sufficient to show that for any concave $R : [0,1] \rightarrow \mathbb{R}_+$, the set on which $\phi_R$ can be discontinuous has zero measure:

\begin{corollary}
Suppose that $R : [0,1] \rightarrow \mathbb{R}_+$ is concave, and not identically zero. Then the set
$$ \hat{E}(R) = \{ (x,y) \in [0,1]^2 \, |\,  R(x) \cdot (1-y) = 2 R(y) \cdot (1-x) \} $$
has Lebesgue measure zero.
\end{corollary}

\begin{proof}
We will show that $E(R) = \hat{E}(R) \cap [0,1)^2$ has Lebesgue measure zero, from which the result follows. 
Note that $E(R)$ is the set of $(x,y)$ such that $F^{-1}(x) = 2 \cdot F^{-1}(y)$. We first show that $E(R)$ is a \emph{function} 
mapping each $x$ to such $y$.\footnote{Set-theoretically, a function is a set containing all input-output pairs.} Let $x,y,y'$ be 
such that $y < y'$, and $F^{-1}(x) = 2 \cdot F^{-1}(y) =  2 \cdot F^{-1}(y')$. Then $F^{-1}(y) =  F^{-1}(y')$. Thus, by 
Lemma~\ref{lem:increase}, $F^{-1}$ is constant on $[y,1]$. 

Now observe that, since $F^{-1}(x) = 2 \cdot F^{-1}(y)$, we have $F^{-1}(x) \geq F^{-1}(y)$. If $x < y$ then $F^{-1}(x) \leq F^{-1}(y)$, by Lemma~\ref{lem:increase}. Therefore, $F^{-1}(x) = F^{-1}(y)$. If instead $x \geq y$, then as $F^{-1}$ is constant on $[y,1]$, we again have $F^{-1}(x) = F^{-1}(y)$. 
However, as $F^{-1}(x) = 2 \cdot F^{-1}(y)$, this implies that $F^{-1}(y) = 0$. Because $F^{-1}$ is constant on $[y,1]$ we have $F^{-1}(1) = 0$. 
So, as $F^{-1}$ is weakly increasing, $F^{-1}$ is identically zero on $[0,1]$. Therefore $R$ is identically zero on $[0,1]$, contradiction.

Consequently, for each $x \in [0,1]$, there exists at most one $y \in [0,1]$ such that $(x,y) \in E(R)$, that is, $E(R)$ is indeed a function. 
Let $d$ denote the domain of $E(R)$. By Lemma~\ref{lem:increase}, $E(R)$ is increasing on its domain. Note that $R(x) / (1-x)$ is 
a continuous function, so $d$ is in fact an interval. Let $a = \inf d$ and $b = \sup d$. Let $\underline{y} = \lim_{x \downarrow a} E(R)(x)$
and let $\bar{y} = \lim_{x \uparrow b} E(R)(x)$. For each $n$, divide $a = x_0$ and $x_i = a + i \cdot 2^{-n} (b-a)$, for $i \leq 2^n$. 
Also let $y_0 = \underline{y}$, $y_{2^n} = \bar{y}$ and $y_i = E(R)(x_i)$, for $1 \leq i < 2^n$. 
Then $E(R) \subseteq \cup_{i = 1}^{2^n} [x_{i-1},x_i] \times [y_{i-1},y_i] = E_n(R)$ for each $n \in \mathbb{N}$. 
Furthermore, $E_n(R)$ has measure $0$ in the limit $n \rightarrow \infty$, as $\mu(E_{n+1}(R)) = 1/2 \cdot \mu(E_n(R))$ for 
any $n \in \mathbb{N}$. 
\end{proof}

As an immediate result, we have the Riemann integrability of $\phi_R$.

\begin{corollary}
Suppose that $R : [0,1]^2 \rightarrow \mathbb{R}_+$ is concave. Then $\phi_R$ is Riemann integrable.
\end{corollary}

\begin{proof}
Since $R$ is continuous, the function is continuous on $[0,1]^2 \setminus E(R)$. In particular, the function is 
discontinuous only on the subset of a set of Lebesgue measure zero. Therefore, by Lebesgue's criterion for 
Riemann integribility, the function is Riemann integrable. 
\end{proof}

So, the integral expression we have makes sense for any concave $R : [0,1]^2 \rightarrow \mathbb{R}_+$. Next, we would like to obtain 
a sequence of probability density functions $f_n$, such that their associated revenue curves $R_n$ tend to $R$. We will do so 
by first defining the following:

\begin{definition}
Let $R : [0,1] \rightarrow \mathbb{R}_+$ be a concave function. Then the $n$-interval piecewise approximation of $R$ is the 
function $R^{|n} : [0,1] \rightarrow \mathbb{R}_+$ such that:
\begin{enumerate}
\item If $x = k / n$ for some $k \in \mathbb{N} \cup \{0\}$ and $k \leq n$, then $R^{|n} = R(x)$.
\item If $x = \lambda k / n + (1-\lambda) (k+1) / n$ for some $k \in \mathbb{N} \cup \{0\}$ and $k \leq n, \lambda \in (0,1)$, then 
$$R^{|n}(x) = \lambda R(k/n) + (1-\lambda) R\left((k+1)/n\right).$$
\end{enumerate}
\end{definition}

We are now ready to prove our convergence theorem:

\begin{proof}
(of Theorem \ref{thm:convergence}) If $R$ is identically zero then the result immediately holds. So suppose not. Note 
that each $R^{|n}$ is again concave, and let $R_n(q) = R^{|2^n}(q) + 1/n \cdot q (1-q)$. We first show that $R_n$ 
converges uniformly to $R$. Let $\epsilon > 0$. If $x = k / 2^m$ for some $k,m$ then, for sufficiently large $n$, we have $R^{|2^n}(x) = R(x)$. If 
not, as $R$ is concave on a compact set, $R$ is uniformly continuous. Hence $\exists \delta > 0, \forall x,y, |x-y| < \delta \Rightarrow |R(x) - R(y)|<\epsilon/2$. Then if $n > -\ln(\delta)$, there exists $k$ such that $k/2^n, (k+1)/2^n$ are in the $\delta$-neighbourhood of $x$. 
Therefore, $\exists \lambda \in (0,1)$ such that $x = \lambda \cdot k/2^n + (1-\lambda) \cdot (k+1)/2^n$. So, by the definition of $R^{|2^n}$,
\begin{align*}
|R^{|2^n}(x)-R(x)| & = | \lambda R(2^{-n}k) + (1-\lambda) R(2^{-n}(k+1)) - R(x) | \\
& \leq \lambda \cdot | R(2^{-n}k) - R(x) | + (1-\lambda) \cdot | R(2^{-n}(k+1)) - R(x) | \\
& < \epsilon / 2.
\end{align*}
Also for sufficiently large $n$, we have $2^{-n} \cdot x(1-x) \leq 2^{-n} \cdot \frac{1}{4} \leq \epsilon / 2$. Together this implies, for 
large $n$, that $|R_n(x) - R(x)| < \epsilon$. As $x$ was arbitrary, we indeed have $R_n$ converging uniformly $R$.

That $R_n$ corresponds to a distribution is immediate as $F_n^{-1}$ is differentiable everywhere on $(0,1)$ except a finite set. 
Furthermore, $F_n^{-1}$ is a bijection by Lemma~\ref{lem:increase}. Thus, $F_n^{-1}$ has a differentiable inverse, $F_n$. In fact, 
for each $s$ on the image $F_n^{-1}\left((0,1)\right)$, if $F(s) = q$ such that $F_n^{-1}(q) = s$, then $f(s) = \left( dF^{-1}_n(q) /dq \right)^{-1}$.
\end{proof}

As a corollary, we conclude that the equivalence (\ref{eqn:equiv}) does hold, showing that our decision to work in quantile space does 
not come at any cost regarding tightness of upper and lower bounds on $\alpha$:

\begin{corollary}
\begin{equation*}
\inf_{F | F \text{ is regular}} \bar{r}_F / r_F  \ =\ \inf_{R | R : [0,1] \rightarrow \mathbb{R}_+ \text{ is concave}} \bar{r}_R / r_R
\end{equation*}
\end{corollary}

\begin{proof}
First let's show LHS $\leq$ RHS. Let $R$ be concave and non-negative on $[0,1]$ and let $R_n$ be as in the proof of Theorem \ref{thm:convergence}. 
Since $R_n$ converges pointwise to $R$, $\phi_{R_n}$ also converges pointwise to $\phi_R$ almost everywhere. Also, 
both $\phi_{R_n}$ and $\phi_R$ are bounded above by the constant function $\sup_{q \in [0,1]} R(q) + 1/2$. Therefore, by the Lebesgue 
dominated convergence theorem,
\begin{equation*}
\int_{(x,y) \in [0,1]^2} \phi_R(x,y) = \lim_{n\rightarrow \infty} \int_{(x,y) \in [0,1]^2} \phi_{R_n}(x,y) = \lim_{n\rightarrow \infty} \bar{r}_{R_n}.
\end{equation*}
However, as $R_n$ is the revenue curve of a distribution which admits a probability density function, we have $\bar{r}_{R_n} = \bar{r}_{F_n}$. 
Further, by definition of $R_n$, for any $\epsilon$ there exists sufficiently large $n$ such that $r_R - \epsilon \leq r_{R_n} \leq r_R + \epsilon$. 
Together, this implies that for any concave and non-negative curve $R$ on $[0,1]$, there exists a sequence of distributions $F_n$ whose 
revenue efficiency converges to $\bar{r}_R$.

Next let's show LHS $\geq$ RHS. Take any regular distribution $F$. By Lemma~\ref{lem:increase}, $F$ can have a point mass only at the supremum 
of its support. Let $\rho$ be this supremum, and let $f(\rho)$ denote the probability the buyer's valuation is $\rho$. Note it is possible 
that $\rho = \infty$, but by regularity of $F$ that implies $f(\rho) = 0$. In this case, note that the following change of variables still works:
\begin{equation*}
\int_{(s,t) \in [0,\rho)} \psi_F(s,t) \cdot dF(s) \times dF(t) \ =\ \int_{(x,y) \in [0,1-f(\rho)] ^2} \phi_R(x,y) \cdot d(x,y).
\end{equation*}
If $f(\rho) = 0$, then we are done. Else, we show that the remaining contributions of $\psi_F$ upper bound the contributions from $\phi_R$. 
If both samples equal $\rho$, then:
\begin{equation*}
\psi_F(\rho,\rho) \ =\ \phi_R(1-f(\rho),1-f(\rho)) \ \geq\  \phi_R(x,y) \qquad \forall x,y \in [1-f(\rho),1]
\end{equation*}
On the other hand, if one sample is $\rho$ and the other sample $t < \rho$, there are two possibilities. If $2t > \rho$, 
then $\psi_F(\rho,t) = \phi(x,F(t))$ for any $x \in [1-f(\rho)]$. If instead $2t < \rho$, then $\psi(\rho,t) = R(1-f(\rho)) \geq \phi(x,F(t))$ 
for any $x \in [1-f(\rho)]$. We conclude that $\bar{r}_F \geq \bar{r}_R$.
\end{proof}

\section{Approximation Programs}\label{sec:prog}

The Riemann integrability of $\phi_R$ on $[0,1]^2$ suggests a possible optimization formulation for our problem. Given a gauge, 
we can try to find a concave and non-negative function $R$ on $[0,1]$, suitably constrained, such that an approximation of $\bar{r}_R$ is 
minimized. Here, we derive the forms of the optimization programs we evaluate, and prove their approximation properties. 
In Section \ref{sec:mainVar} we present our main primal variables by formulating a class of quadratically-constrained programs 
whose solutions minimize approximations of (\ref{eqn:qSpace}) given some gauge on $[0,1]$. In Section \ref{sec:progUB} we provide a quadratic 
objective function for (\ref{eqn:qSpace}). Linearisation of the objective and the constraints then provides a family of MILPs suitable for searching for 
minimal distributions for the ERM mechanism with two samples, while proofs of well-behaviour of feasible solutions show that the value of these MILPs 
converge to $\alpha$ under gauge refinements. Motivated by this in Section \ref{sec:progLB} 
we formulate a cubic objective for (\ref{eqn:qSpace}) which allows us to obtain a family of MILPs whose values provide lower bounds on $\alpha$. 
In Section \ref{sec:perf}, we detail several considerations we employ to ensure that our MILPs are practically solvable and provide good bounds.

\subsection{Primal Variables}\label{sec:mainVar}

To compute a Riemann sum of $\phi_R$ on $[0,1]^2$, we would first need to define a gauge on $[0,1]^2$. Here, we opt for the natural approach, 
defining a gauge on $[0,1]^2$ by considering product intervals arising from a gauge on $[0,1]$. In this line, suppose we divide the interval $[0,1]$ 
into subintervals of the form $I(i) = [q_i,q_{i+1}]$ for $1 \leq i \leq n$, where $q_1 = 0$, $q_{n+1} = 1$, 
and $q_{i+1} > q_i$ for any $1 \leq i \leq n$. Also denote by $I(i,j)$ the product interval $[q_i,q_{i+1}]\times[q_j,q_{j+1}]$. Then we may rewrite
integral (\ref{eqn:qSpace}) as:
\begin{align}
\bar{r}_R & = \int_{(x,y) \in [0,1]^2} \phi_R(x,y) \cdot d(x,y) \label{eqn:integral} \\
& = \sum_{1 \leq i,j \leq n} \int_{(x,y) \in I(i,j)} \phi_R(x,y) \cdot d(x,y) \nonumber \\
& = \sum_{1 \leq i \leq n} \int_{(x,y) \in I(i,i)} \phi_R(x,y)\cdot d(x,y) + 2 \cdot \sum_{1 \leq j < i \leq n} \int_{(x,y) \in I(i,j)} \phi_R(x,y) \cdot d(x,y) \nonumber
\end{align}

We want primal solutions to our problems to describe approximately minimal value distributions for the buyer. One way to do so is to include 
variables that correspond to the values the revenue curve attains. Specifically, for $(q_i)_{1 \leq i \leq n+1}$, we will include variables $R(q_i)$. 
For notational convenience later on, let $\vec{R}$ denote the vector containing all $R(q_i)$.

Then each $R(q_i)$ corresponds to a value attained by a non-negative, concave function. This implies that the following constraints must hold:
\begin{align}
R(q_i) \cdot (q_{i+1} - q_{i-1}) & \geq R(q_{i+1}) \cdot (q_i - q_{i-1}) + R(q_{i-1}) \cdot (q_{i+1}-q_i) && \forall 1 < i < n+1 \label{con:Rconcave}\\
R(q_i) & \geq 0 && \forall 1 \leq i \leq n+1
\end{align}

Furthermore, we want $R$ to be normalized such that $\max_{q \in [0,1]} R(q) = 1$. Unfortunately, this is non-trivial to implement linearly. So, 
instead, we constrain the set of revenue curves so that there exists $1 \leq OPT \leq n$ and $q^* \in [q_{OPT},q_{OPT+1}]$ 
such that $R(q^*) = 1$. By the concavity and non-negativity of $R$, this implies that:
\begin{align}
R(q_{OPT}) & \geq q_{OPT} / q_{OPT+1} \label{con:Ropt1}\\
R(q_{OPT+1}) & \geq (1-q_{OPT+1}) / (1-q_{OPT}) \label{con:Ropt2}
\end{align}
Furthermore, by concavity, $R(\cdot)$ should be weakly increasing before $q_{OPT}$ and weakly 
decreasing beyond $q_{OPT+1}$:
\begin{align}
R(q_{i+1}) - R(q_i) &\leq 0 \qquad \forall 1 \leq i < q_{OPT} \label{con:Ropt3} \\
- R(q_{i+1}) + R(q_i) &\leq 0 \qquad \forall q_{OPT+1} \leq i < n \nonumber
\end{align}

We also model the indicator functions in (\ref{def:phiR}) as binary variables:

\begin{lemma}\label{lem:wCons}
For any $(x,y) \in [0,1)$ such that $x > y$,
\begin{align}
\phi_R(x,y) = & \min_{w(x,y) \in \{0,1\}} &  R(x) \cdot w(x,y) + R(y) \cdot (1-w(x,y))  & \nonumber \\
& \text{subject to} & w(x,y) \cdot [ R(x) \cdot (1-y) - 2 R(y) \cdot (1-x) ] & \geq 0 \label{con:wx}\\
&  & (1-w(x,y)) \cdot [ R(x) \cdot (1-y) - 2 R(y) \cdot (1-x) ] & \leq 0 \label{con:wy}
\end{align}
\end{lemma}

\begin{proof}
If $R(x) (1-y) > 2 R(y) (1-x)$ then the only feasible point is $w(x,y) = 1$. This correctly sets $\phi_R(x,y) = R(x)$. 
Likewise, if $R(x) (1-y) < 2 R(y) (1-x)$ then the 
only feasible point is $w(x,y) = 0$, which correctly sets $\phi_R(x,y) = R(y)$. If $R(x) (1-y) = 2 R(y) (1-x)$, then both $w(x,y) = 1$ and $w(x,y) = 0$ are 
feasible. This correctly sets $\phi_R(x,y) = \min \{ R(x),R(y) \}$.
\end{proof}

To compute a Riemann sum, we evaluate $w$ on a set $T$ of points in $[0,1]^2$ such that:
$$ \forall 1 \leq j < i \leq n,\quad \exists (\bar{q}_i,\bar{q}_j) \in T,\  (\bar{q}_i,\bar{q}_j) \in [q_i,q_{i+1}] \times [q_j,q_{j+1}].$$
This condition implies that each non-diagonal area element $I(i,j)$ for $1 \leq j < i \leq n$ contains a point where we evaluate $w$. 

We include variables for the value $R$ attains on endpoints of intervals, but $w$ may be evaluated (in principle) anywhere on $I(i,j)$. 
Then for $(x,y) \in T$, to be able to impose constraints of the form (\ref{con:wx}) and (\ref{con:wy}) on $w(x,y)$, we find $R(x)$ and $R(y)$ by 
linear interpolation on $\vec{R}$. In particular, if $x \in [q_i,q_{i+1}]$, then:
$$ R(x) \cdot (q_{i+1}-q_i) = R(q_i) \cdot (q_{i+1}-x) + R(q_{i+1}) \cdot (x - q_i),$$
and likewise for $R(y)$. So setting $\vec{w}$ to be the vector containing all $w(x_\ell,y_\ell)$, for each individual summand in (\ref{eqn:integral}) we 
may approximate
\begin{equation*}
\int_{(x,y) \in I(i,j)} \phi_R(x,y) \cdot d(x,y) \simeq A(i,j) \cdot f_{ij}(\vec{R},\vec{w})
\end{equation*}
where $A(i,j) = (q_{i+1}-q_i)(q_{j+1}-q_j)$ is the area of $I(i,j)$, for $1 \leq j \leq i \leq n$, and $f_{ij}$ is some function determined by our 
approximation scheme, homogeneous of degree one in $\vec{R}$.

This provides the form of our most general optimization formulation: we consider a set of gauges indexed by a set $J$, $(\vec{q}^k)_{k \in J}$, such 
that $\cup_{k \in J} [q^k_{OPT^k}, q^k_{OPT^k+1}] = [0,1]$, and find $R$ that minimizes our approximation of $\bar{r}_R$ by computing:
\begin{align}
\min_{k \in J} \min_{\vec{R},\vec{w}} & \quad \sum_{1 \leq i \leq n} A_k(i,i) \cdot f_{ii}(\vec{R},\vec{w}) 
	+ 2 \cdot \sum_{1 \leq j < i \leq n} A_k(i,j) \cdot f_{ij}(\vec{R},\vec{w})  \label{opt:general} \\
\text{subject to} &  \quad (\ref{con:Rconcave}),(\ref{con:Ropt1}),(\ref{con:Ropt2}),(\ref{con:Ropt3}),(\ref{con:wx}),(\ref{con:wy}) \nonumber \\
& \quad \vec{R} \in [0,1]^{n+1} \nonumber \\ 
& \quad \vec{w} \in \{0,1\}^{T} \nonumber
\end{align}

\subsection{Upper Bound: A Quadratic Formulation}\label{sec:progUB}

To derive an upper bound, we will need to find an approximately-minimal revenue curve. We consider a straightforward 
implementation of (\ref{opt:general}) to do this. For $n \in \mathbb{N}$, we take the uniform gauge given by $q_i = (i-1) /n$ 
for $1 \leq i \leq n+1$, and consider each case when the peak of the revenue curve is in $[q_k,q_{k+1}]$ for $1 \leq k \leq n$. 
To evaluate the Riemann sum, mark the midpoint of each interval:
$$ \bar{q}_i = \frac{q_i + q_{i+1}}{2} $$
Then to approximate our Riemann integral, for each $I(i,j)$ we will evaluate the function $\phi_R$ at $(\bar{q}_i,\bar{q}_j)$. So we set:
\begin{align*}
f_{ii} & = \frac{R(q_i)+R(q_{i+1})}{2} & & \forall 1 \leq i \leq n \\
f_{ij} & = \frac{R(q_i)+R(q_{i+1})}{2} \cdot w(\bar{q}_i,\bar{q}_j) + \frac{R(q_j)+R(q_{j+1})}{2} \cdot (1-w(\bar{q}_i,\bar{q}_j)) & & \forall 1 \leq j < i \leq n
\end{align*}

If we evaluate the resulting optimization problem, the constraints (\ref{con:Ropt1}) and (\ref{con:Ropt2}) tend to ``chip off'' the 
peak of the revenue curve in the primal solutions. This is unlikely to be a feature of an actual minimal revenue curve, so we will 
convert the constraints (\ref{con:Ropt1}) and (\ref{con:Ropt2}) into a single equality constraint, at the cost of increasing the size 
of the index set $J$ by one. First observe that 
$$\max \{ R(q_{OPT}), R(q_{OPT+1}) \} \ \geq\  \max \{q_{OPT} / q_{OPT+1}, (1-q_{OPT+1}) / (1-q_{OPT}) \}$$
for any feasible solution $(\vec{R},\vec{w})$. Let $\vec{R}^* = \max \{ R(q_{OPT}), R(q_{OPT+1})\}^{-1} \vec{R}$. Then for any $1 \leq j \leq i \leq n$, by the 
homogeneity of $f_{ij}$ in $\vec{R}$:
\begin{align}
A(i,j) \cdot f_{ij}(\vec{R},\vec{w}) 
	& = \max \{ R(q_{OPT}), R(q_{OPT+1}) \} \cdot A(i,j) \cdot f_{ij}(\vec{R}^*,\vec{w}) \label{eqn:loss}\\
& \geq \max \Bigg\{\frac{q_{OPT}}{q_{OPT+1}}, \frac{1-q_{OPT+1}}{1-q_{OPT}} \Bigg\} \cdot A(i,j) \cdot f_{ij}(\vec{R}^*,\vec{w})  \nonumber\\
& \geq \frac{n-1}{n+1} \cdot A(i,j) \cdot f_{ij}(\vec{R}^*,\vec{w}) \nonumber
\end{align}
Now, $\vec{R}^*$ has either $R(q_{OPT}) = 1$ or $R(q_{OPT+1}) = 1$. So we consider imposing such an equality 
constraint in our optimization programs to normalise the maximum of the revenue curve, dropping the optimality 
constraints (\ref{con:Ropt1}) and (\ref{con:Ropt2}) from our optimization program. We are also able to drop the 
constraint (\ref{con:Ropt3}), since it is implied by $R(q_k) = 1$, $R \leq 1$, and the concavity constraints (\ref{con:Rconcave}).

Finally, note that with the uniform gauge, $A(i,j) = 1/n^2$, for any $1  \leq j \leq i \leq n$. Thus, our Riemann sum 
minimization program is:
\begin{align}
\hat{\alpha}(n) = \min_{1 \leq k \leq n+1} \min_{\vec{R},\vec{w}} & \quad \sum_{1 \leq i \leq n} \frac{1}{n^2} \cdot f_{ii}(\vec{R},\vec{w}) 
	+ 2 \cdot \sum_{1 \leq j < i \leq n} \frac{1}{n^2} \cdot f_{ij}(\vec{R},\vec{w})  \label{opt:approx} \\
\text{subject to} &  \quad (\ref{con:Rconcave}),(\ref{con:wx}),(\ref{con:wy}) \nonumber \\
& \quad R(q_k) = 1 \nonumber \\
& \quad \vec{R} \in [0,1]^{n+1} \nonumber \\ 
& \quad \vec{w} \in \{0,1\}^{\binom{n}{2}} \nonumber
\end{align}

Intuitively, since the factor $(n-1)/(n+1)$ in (\ref{eqn:loss}) goes to $1$  as $n$ grows large, this program should be able to approximate $\alpha$: 

\begin{theorem}\label{thm:easyApprox}
As $n \rightarrow \infty$, $\hat{\alpha}(n) \rightarrow \alpha$.
\end{theorem}

Before we prove this theorem, we emphasize an important \emph{monotonicity} property of $w(x,y)$: it is non-decreasing in 
the first argument and non-increasing in the second argument:

\begin{lemma}\label{lem:wMon}
Suppose that $R$ is concave and non-negative on $[0,1]$, and $w$ is determined as in Lemma~\ref{lem:wCons}. Then 
for any $x,y \in [0,1]^2$ such that $x > y$:
\begin{enumerate}
\item[(i)] If $x' > x$, then $w(x',y) \geq w(x,y)$. 
\item[(ii)] If $x > y' > y$, then $w(x,y') \leq w(x,y)$.
\end{enumerate}
\end{lemma} 

\begin{proof}
Now (i) can equivalently be stated $w(x,y) = 1 \Rightarrow w(x',y) = 1$, which we prove. Note that $R(x) / (1-x)$ is non-decreasing on $[0,1)$, 
which implies that $R(x) \cdot (1-y) - 2 R(y) \cdot (1-x)$ is non-decreasing in $x$ on $[0,1]$. If it is the case that the 
constraint (\ref{con:wx}) does not bind or if $R(x) \cdot (1-y) - 2 R(y) \cdot (1-x)$ strictly increases, then we are done. Else, it 
must be that $R(x') = (1-x') R(x) / (1-x) < R(x)$ and so $\min\{ R(x'), R(y) \} = R(x')$. Therefore, $w(x',y) = 1$. Likewise, (ii)
can be equivalently stated as $w(x,y) = 0 \Rightarrow w(x,y') = 0$. Because $R(y) / (1-y)$ is increasing in $y$, we have 
$R(x) \cdot (1-y) - 2 R(y) \cdot (1-x)$ is 
non-increasing in $y$ on $[0,1]$. From this, (ii) follows by an analogous argument.
\end{proof}

These monotonicity properties of $w$ imply that only few $w(\bar{q}_i,\bar{q}_j)$'s may be ``misspecified''. In particular, for some revenue curve $R$, 
the objective contributions $A(i,j) \cdot f_{ij}(\vec{R},\vec{w})$ all underestimate their corresponding terms in \ref{eqn:integral} except for a vanishing fraction of product intervals $I(i,j)$:

\begin{lemma}\label{lem:fewError}
Let $(\vec{R},\vec{w})$ be a feasible solution of (\ref{opt:approx}), and let $R$ be a revenue curve agreeing with $\vec{R}$ on 
the gauge $(q_i)_{1 \leq i \leq n+1}$. Then for at least $\binom{n}{2} - 2n + 3$ many pairs $(i,j)$ such that $1 \leq j < i \leq n$, $w$ is a 
constant function on $I(i,j)$. In particular, for such pairs $(i,j)$:
$$ \int_{(x,y) \in I(i,j)} \phi_R(x,y) \cdot d(x,y) \geq A(i,j) \cdot \Bigg[ \frac{R(q_i)+R(q_{i+1})}{2} \cdot w(\bar{q}_i,\bar{q}_j) + \frac{R(q_j)+R(q_{j+1})}{2} \cdot (1-w(\bar{q}_i,\bar{q}_j)) \Bigg],$$
with equality if $R$ is the linear interpolation of $\vec{R}$.
\end{lemma}

\begin{proof}
We know that $w$ is determined as in Lemma~\ref{lem:wCons}, that $w$ and $\vec{w}$ may be taken to agree 
on $(\bar{q}_i,\bar{q}_j)_{1 \leq j < i \leq n}$, and that $w$ satisfies monotonicity by Lemma~\ref{lem:wMon}. Now, 
extend $w$ to points of the form $(q_i,q_j)_{1 \leq j \leq i \leq n+1}$ by setting $w(q_i,q_i) = 0$ for any $1 \leq i \leq n$. 
We remark that this extension of $w$ will still satisfy monotonicity.

We now define a notion of \emph{constantness} for $w$ on any $I(i,j)$ with $j < i$. We will say that the pair $(i,j)$ 
is \emph{$1$-definite} if $w(q_{i+1},q_j) = w(q_i,q_{j+1})=1$, and \emph{$0$-definite} if $w(q_{i+1},q_j) = w(q_i,q_{j+1})=1$. 
Else, by the monotonicity of $w$, it must be that $w(q_{i+1},q_j) = 1$ and $w(q_i,q_{j+1})=0$; we call such a pair $(i,j)$ \emph{indefinite}. 
Then, by the monotonicity of $w$, it holds that:
\begin{enumerate}
\item If $(i,j)$ is $1$-definite then $(i+1,j-1)$ is $1$-definite.
\item If $(i,j)$ is $0$-definite and $j+1 < i-1$, then $(i-1,j+1)$ is $0$-definite.
\item If $(i,j)$ is indefinite then $(i+1,j-1)$ is $1$-definite, and if also $j+1 < i-1$, then $(i-1,j+1)$ is $0$-definite.
\end{enumerate}
Therefore, for each $3 \leq k \leq 2\cdot n-1$, the set of pairs $\{ (i,j) | 1 \leq j < i \leq n \text{ and } i+j = k\}$ contains at 
most one indefinite pair. There are only $2n-3$ such possible values of $k$. Now, if the pair $(i,j)$ is $1$-definite, then
\begin{equation*}
\int_{(x,y) \in I(i,j)} \phi_R(x,y) \cdot d(x,y) = \int_{(x,y) \in I(i,j)} R(x) \cdot d(x,y) \geq A(i,j) \cdot \frac{R(q_i)+R(q_{i+1})}{2},
\end{equation*}
where the inequality holds due to concavity of $R$. If $R$ is an affine function on $I(i)$, then the inequality in fact holds with 
equality. A similar (in)equality holds if the pair $(i,j)$ is $0$-definite, which implies the result.
\end{proof}

This result allows us to prove Theorem \ref{thm:easyApprox}, and provide within the proof an explicit error estimate for $\hat\alpha(n)$:

\begin{proof}(of Theorem \ref{thm:easyApprox}) For fixed $n \in \mathbb{N}$, let $R$ be the linear interpolation 
of a minimum primal solution of (\ref{opt:approx}). Then,
$$\int_{(x,y) \in I(i,j)} \phi_R(x,y) \cdot d(x,y) = \frac{1}{n^2} \cdot \Bigg[\frac{R(q_i)+R(q_{i+1})}{2} \cdot w(\bar{q}_i,\bar{q}_j) + \frac{R(q_j)+R(q_{j+1})}{2} \cdot (1-w(\bar{q}_i,\bar{q}_j)) \Bigg]$$
for any $1 \leq j < i \leq n$, except at most $2n-3$ many by Lemma \ref{lem:fewError}. Also note that we may be misspecifying the contribution 
of areas of the form $(i,i)$. As $R \leq 1$, if we estimate $\phi_R(x,y) = 1 + f_{ij}(\vec{R},\vec{w})$ on such areas, 
we overestimate $\bar{r}_R$. Therefore, 
$\bar{r}_R \leq (5n-6) / n^2 + \hat{\alpha}(n)$. 
Since $\bar{r}_R \geq \alpha$, by the infimum property of $\alpha$, letting $n$ be sufficiently large such 
that $(5n-6)/n^2 \leq \epsilon$, we have $\hat\alpha(n) \geq \alpha - \epsilon$.

To show the upper bound on $\hat\alpha(n)$, let $\epsilon,\delta > 0$. Let $R_\delta$ be a revenue curve such 
that $\bar{r}_{R_\delta} / r_{R_\delta} \leq \alpha + \delta$, normalised 
such that $\max_{q \in [0,1]} R_\delta(q) = 1$. Set $M(n) = \max_{q \in [0,1]} R_\delta^{|n}(q)$ for any $n$. 
Then $\lim_{n \rightarrow \infty} M(n) = 1$. Also, since $R_\delta$ 
is regular, $M(n) \geq (n-1)/(n+1)$. Towards this end, we choose $n \in \mathbb{N}$ such that
$$ \frac{2}{n-1} + \frac{5n-6}{n^2} < \epsilon.$$

Now, set $R_\delta^n = R^{|n} / M(n)$, the renormalized linear approximation to $R_\delta$. Then notice that an interval product $I(i,j)$ is 
definite for $R_\delta$ if and only if it is definite for $R_\delta^n$. Hence for any definite $I(i,j)$, the Riemann sum contribution of $\phi_{R_\delta}/ M(n)$ is 
greater than that of $\phi_{R_\delta^n}$. If $I(i,j)$ is instead any interval, the Riemann sum contribution of $\phi_{R_\delta^n}$ at $I(i,j)$ is at most $1/n^2$. Therefore,
\begin{align*}
\sum_{1 \leq i,j \leq n} & \frac{1}{n^2} \cdot \phi_{R_\delta^n}(\bar{q}_i,\bar{q}_j) \\ 
& = \frac{2}{n^2} \cdot \Bigg[ \sum_{\substack{1 \leq j < i \leq n \\ I(i,j) \text{ definite}}} \phi_{R_\delta^n}(\bar{q}_i,\bar{q}_j) + 
\sum_{\substack{1 \leq j < i \leq n \\ I(i,j) \text{ indefinite}}} \phi_{R_\delta^n}(\bar{q}_i,\bar{q}_j) \Bigg] + 
\sum_{1 \leq i \leq n} \frac{1}{n^2} \cdot \phi_{R_\delta^n}(\bar{q}_i,\bar{q}_i) \\
& \leq  \sum_{1 \leq i,j \leq n}  \int_{(x,y) \in I(i,j)} \phi_{R_\delta/M(n)}(\bar{q}_i,\bar{q}_j) \cdot d(x,y) + \frac{4n-6}{n^2} + \frac{n}{n^2} \\
& = \frac{n-1}{n+1}  \cdot \sum_{1 \leq i,j \leq n} \int_{(x,y) \in I(i,j)} \phi_{R_\delta}(\bar{q}_i,\bar{q}_j) \cdot d(x,y) + \frac{5n-6}{n^2} \\
& = \frac{n-1}{n+1} \cdot \bar{r}_{R_\delta} + \frac{5n-6}{n^2} \leq \bar{r}_{R_\delta} + \frac{2}{n-1} + \frac{5n-6}{n^2} \leq \alpha + \delta + \epsilon
\end{align*}
Here, the first inequality follows from Lemma \ref{lem:fewError}, the second equality follows from homonegeity of $\phi$ in revenue curves, and the second 
inequality follows since $\bar{r}_{R_\delta} \leq 1$ by normalisation of $R_\delta$. However, 
$$ \sum_{1 \leq i,j \leq n} \frac{1}{n^2} \cdot \phi_{R_\delta^n}(\bar{q}_i,\bar{q}_j) \geq \hat\alpha(n),$$
as the sum equals the value of (\ref{opt:approx}) when we plug in feasible solution $\vec{R} = ( R^n_\delta(q_i) )_{1 \leq i \leq n+1}$. Combining these, we get the inequality
\begin{equation*}
\hat\alpha(n) 
\ \leq\  \sum_{1 \leq i,j \leq n} \frac{1}{n^2} \cdot \phi_{R_\delta^n}(\bar{q}_i,\bar{q}_j) 
\ \leq\ \alpha + \delta + \epsilon.
\end{equation*}
Since our choice of $\delta$ is independent from our choice of $\epsilon$ and $n$, we may take $\delta \downarrow 0$, from which the theorem follows.
\end{proof}

The convergence result suggests a natural optimization scheme to find an approximately minimal distribution -- we linearize the 
terms of the form $R(q_\ell) \cdot w(\bar{q}_i,\bar{q}_j)$ in the objective and the constraints, adding in the constraints from the 
second-order Sherali-Adams lift of (\ref{opt:approx}) that include such terms. In particular, we add in the constraints
\begin{align}
Rw(\ell,i,j) & \geq 0 \label{con:ASapprox} \\
w(\bar{q}_i,\bar{q}_j) - Rw(\ell,i,j) & \geq 0 \nonumber \\
R(\bar{q}_\ell) - Rw(\ell,i,j) & \geq 0 \nonumber \\
-R(\bar{q}_\ell) - w(\bar{q}_i,\bar{q}_j) +Rw(\ell,i,j)  & \geq -1 \qquad \forall \  1 \leq j < i \leq n,\, \ell \in \{i,j\}, \nonumber 
\end{align}
where $Rw(\ell,i,j)$ is a variable representing the product $R(\bar{q}_\ell) \cdot w(\bar{q}_i,\bar{q}_j)$. We then replace the product 
terms in constraints (\ref{con:wx}), (\ref{con:wy}) and in the objective with the corresponding linearized variable.  Note that these 
constraints imply that $Rw(\ell,i,j) = R(\bar{q}_\ell) \cdot w(\bar{q}_i,\bar{q}_j)$ whenever $w(\bar{q}_i,\bar{q}_j)$ is $\{0,1\}$-valued. 
Therefore, the mixed-integer LP formulation is exact.

Finally, we impose the monotonicity constraints, implied by Lemma~\ref{lem:wMon}, even though these constraints 
are \emph{redundant} for our formulation (cf. Section \ref{sec:perf}). This implies that our approximate MILP has the following form:
\begin{align}
\min_{1 \leq k \leq n+1} \min_{\vec{R},\vec{w},Rw} & \quad \sum_{1 \leq i \leq n} \frac{1}{n^2} \cdot R(\bar{q}_i) + \sum_{1 \leq j < i \leq n}\frac{2}{n^2} \cdot (Rw(i,i,j) + R(\bar{q}_j) - Rw(j,i,j)) \label{opt:approx2} \\
\text{subject to} &  \quad (\ref{con:Rconcave}),(\ref{con:wx}),(\ref{con:wy}),\ref{con:ASapprox} \nonumber \\
& \quad w(\bar{q}_i,\bar{q}_j) \leq w(\bar{q}_{i+1},\bar{q}_j) \ \forall \ 1 \leq j < i < n \nonumber \\
& \quad w(\bar{q}_i,\bar{q}_j) \geq w(\bar{q}_i,\bar{q}_{j+1}) \ \forall \ 1 \leq j < i+1 \leq n \nonumber \\
& \quad R(q_k) = 1 \nonumber \\
& \quad \vec{R} \in [0,1]^{n+1} \nonumber \\ 
& \quad \vec{w} \in \{0,1\}^{\binom{n}{2}} \nonumber \\
& \quad Rw \in [0,1]^{2 \times \binom{n}{2}} \nonumber
\end{align}

\subsection{Lower Bound: A Cubic Formulation}\label{sec:progLB}

In this section, we construct an MILP with the explicit aim of obtaining lower bounds for $\alpha$. The MILP (\ref{opt:approx2}) does 
provide certifiable lower bounds for $\alpha$ by Theorem \ref{thm:easyApprox}. However, the exponential nature of the problem kicks in 
before we can certify any significant improvement on the lower bound of $\simeq .558$ provided Daskalakis and Zampetakis~\cite{DZ20}.

We will work around this problem by considering a cubic program which, given a gauge, lower bounds the contribution of any area element. 
Handling the contribution of area elements on the diagonal will be straightforward, but the contributions from the off-diagonal area 
elements will require care. Towards this end, we will use Lemma~\ref{lem:wMon}, and lower bound this contribution conditional on the 
definiteness of the area element.

To construct such a lower bound program, we first need to fix our gauges: for the general formulation of the problem (\ref{opt:general}), 
we find a set of gauges $(\vec{q}^k)_{k \in I}$ with prescribed \emph{optimal intervals} $[q^k_{OPT^k},q^k_{OPT^k+1}]$ such 
that $\cup_{k \in I} [q^k_{OPT^k}, q^k_{OPT^k+1}] = [0,1]$. As evidenced by (\ref{eqn:loss}), we will want the freedom to 
pick $q^k_{OPT^k + 1} -  q^k_{OPT^k}$ small for each gauge $\vec{q}^k$ to minimize the loss from relaxing the optimality 
constraint to (\ref{con:Ropt1}) and (\ref{con:Ropt2}). To this end, for some $N \in \mathbb{N}$ ``significantly larger'' than $n$, we 
will set $J = \{1,2,...,N\}$ and:
\begin{align*}
q^k_{OPT^k} & = \frac{k-1}{N} \\
q^k_{OPT^k+1} & = \frac{k}{N}.
\end{align*}
Then by (\ref{eqn:loss}), we expect degredations on the quality of the lower bound caused by the optimality constraints to be of 
order $\sim 1/N$ as we impose larger $N$. Note that this only comes at a linear cost of having to compute $N$ MILPs.

Next, we need to decide on where to evaluate each $w(q_i,q_j)$. By Lemma \ref{lem:wMon}, to decide on the definiteness of an area 
element $[q_i,q_{i+1}] \times [q_j,q_{j+1}]$ for $1 \leq j < i \leq n$, we need to check $w(q_{i+1},q_{j})$ and $w(q_i,q_{j+1})$. 
Due to this constraint, we also need to assign a value to $w$ on $(q_i,q_i)_{1 \leq i \leq n+1}$. The defining constraints (\ref{con:wx}) 
and (\ref{con:wy}) become degenerate on such points. Instead we will opt to always fix $w(q_i,q_i) = 0$, as such an assignment 
respects monotonicity and we wish to avoid adding even more binary variables.

We are now ready to derive lower bounds on the contribution of each area element. As promised, lower bounding the contribution 
of a diagonal area element is simple:

\begin{lemma}\label{lem:diagbound}
Suppose $R : [0,1] \rightarrow \mathbb{R}_+$ is concave and $I^k(OPT^k) \cap \arg \max_{q \in [0,1]} R(q)$ is non-empty. 
Then the following hold:
\begin{enumerate}
\item[(i)] If $i < OPT^k$, then
$\int_{(x,y) \in I^k(i,i)} \phi_R(x,y)\cdot d(x,y) \geq A_k(i,i) \cdot \Bigg( \frac{2R(q^k_i)}{3} + \frac{R(q^k_{i+1})}{3} \Bigg) $.
\item[(ii)] If $i = OPT^k$, then
$ \int_{(x,y) \in I^k(i,i)} \phi_R(x,y)\cdot d(x,y) \geq 0$.
\item[(iii)] If $i > OPT^k$, then
$ \int_{(x,y) \in I^k(i,i)} \phi_R(x,y)\cdot d(x,y) \geq A_k(i,i) \cdot \Bigg( \frac{R(q^k_i)}{3} + \frac{2R(q^k_{i+1})}{3} \Bigg) $.
\end{enumerate}
\end{lemma}

\begin{proof}
Observe that (ii) is trivial, so we show (i) and (iii). Let $R'$ be the concave function on $[0,1]$, obtained by linear interpolation 
on $(R(q_i))_{1 \leq i \leq n+1}$. Then $R \geq R'$ on $[0,1]$, which implies that:
$$ \int_{(x,y) \in I^k(i,i)} \phi_R(x,y)\cdot d(x,y) \geq \int_{(x,y) \in I^k(i,i)} \min \{ R'(x), R'(y) \} \cdot d(x,y).$$
$R'$ is non-decreasing on $[0,q^k_{OPT^k}]$ and non-increasing on $[q^k_{OPT^k+1},1]$, and is affine on any $I^k(i)$. 
Thus for a lower bound for the revenue contribution from the product interval $I^k(i,i)$, we evaluate the minimum of two independent 
variables drawn from the uniform distribution on $R'(I^k(i))$, which implies that (i) and (iii) hold.
\end{proof}

Next, the lower bounds on off-diagonal area elements, in turn, are dependent on whether the area element is definite:

\begin{lemma}\label{lem:nondiagbound}
Suppose $R : [0,1] \rightarrow \mathbb{R}_+$ is concave, $1 \leq j < i \leq n$, and $[q^k_{OPT^k},q^k_{OPT^k+1}] \cap \arg \max_{q \in [0,1]} R(q)$ 
is non-empty. Then the following hold:
\begin{enumerate}
\item[(a)] If the pair $(i,j)$ is $1$-definite, then:
$$ \int_{(x,y) \in I(i,j)} \phi_R(x,y) \cdot d(x,y) \geq A_k(i,j) \cdot \frac{R(q^k_{i}) + R(q^k_{i+1})}{2}.$$
\item[(b)]  If the pair $(i,j)$ is $0$-definite, then:
$$ \int_{(x,y) \in I(i,j)} \phi_R(x,y) \cdot d(x,y) \geq A_k(i,j) \cdot \frac{R(q^k_{j}) + R(q^k_{j+1})}{2}.$$
\item[(c)]  If the pair $(i,j)$ is indefinite and $i < OPT^k$, then:
$$ \int_{(x,y) \in I(i,j)} \phi_R(x,y) \cdot d(x,y) \geq A_k(i,j) \cdot \frac{R(q^k_{j}) + R(q^k_{j+1})}{2}.$$
\item[(d)]  If the pair $(i,j)$ is indefinite and $j > OPT^k$, then:
$$ \int_{(x,y) \in I(i,j)} \phi_R(x,y) \cdot d(x,y) \geq A_k(i,j) \cdot \frac{R(q^k_{i}) + R(q^k_{i+1})}{2}.$$
\item[(e)]  If the pair $(i,j)$ is indefinite, $i \geq OPT^k$ and $j \leq OPT^k$, then
$$ \int_{(x,y) \in I(i,j)} \phi_R(x,y) \cdot d(x,y) \geq A_k(i,j) \cdot \mathbb{E}[\min\{\underline{R}(x),\underline{R}(y)\}|(x,y) \in I(i,j)],$$
where $\underline{R}$ is the minimum concave and non-negative function on $[0,1]$ satisfying (\ref{con:Ropt1}) and (\ref{con:Ropt2}).
\end{enumerate}
\end{lemma}

\begin{proof}
If the pair $(i,j)$ is $1$-definite, then
\begin{align*}
\int_{(x,y) \in I(i,j)} \phi_R(x,y) \cdot d(x,y) & = A_k(i,j) \cdot \mathbb{E}[R(x) | q^k_i \leq x \leq q^k_{i+1}] \\
& \geq A_k(i,j) \cdot \frac{R(q^k_{i}) + R(q^k_{i+1})}{2},
\end{align*}
where the equality holds since $w(x,y) = 1$ on $I(i,j)$ and the inequality holds by the concavity of $R$. This gives (a). 
The case (b) when $(i,j)$ is $0$-definite holds similarly.

Now suppose that the pair $(i,j)$ is indefinite. Then we use the inequality
$$\int_{(x,y) \in I(i,j)} \phi_R(x,y) \cdot d(x,y) \geq A_k(i,j) \cdot \mathbb{E}[\min\{R(x),R(y)\}|(x,y) \in I(i,j)].$$
If $i < OPT_k$, then $R(y) \leq R(x)$ for $(x,y) \in I(i,j)$ as $R$ is concave and obtains its maximum at a 
point $q^* \geq q^k_{OPT^k} \geq x \geq y$. This implies (c) as:
\begin{align*}
 A_k(i,j) \cdot \mathbb{E}[\min\{R(x),R(y)\}|(x,y) \in I(i,j)] & = A_k(i,j) \cdot \mathbb{E}[R(y)|q^k_j \leq y \leq q^k_{j+1}] \\
& \geq A_k(i,j) \cdot \frac{R(q^k_{j}) + R(q^k_{j+1})}{2}.
\end{align*}
The case (d) where $j > OPT_k$ follows analogously. If neither case holds, then by the concave closure 
property we have $\underline{R} \leq R$. Therefore,
$$A_k(i,j) \cdot \mathbb{E}[\min\{R(x),R(y)\}|(x,y) \in I(i,j)] \geq A_k(i,j) \cdot \mathbb{E}[\min\{\underline{R}(x),\underline{R}(y)\}|(x,y) \in I(i,j)].$$ 
Thus (e) holds.
\end{proof}

This allows us to write a cubic expression which lower bounds the contribution from an off-diagonal area element to the revenue:

\begin{corollary}\label{cor:nondiagbound}
Suppose $R : [0,1] \rightarrow \mathbb{R}_+$ is concave, $1 \leq j < i \leq n$, 
and $[q^k_{OPT^k},q^k_{OPT^k+1}] \cap \arg \max_{q \in [0,1]} R(q)$ is non-empty. 
Let $f_{1ij}(\vec{R}),f_{0ij}(\vec{R}),f_{\iota ij}(\vec{R})$ be respectively the lower bounds on the revenue 
contribution from the area element $I(i,j)$, conditional respectively on the pair $(i,j)$ being $1$-definite, $0$-definite 
or indefinite as in Lemma~\ref{lem:nondiagbound}. Then:
\begin{align*}
\int_{(x,y) \in I(i,j)} \phi_R(x,y) \cdot d(x,y) & \geq A_k(i,j)f_{1ij}(\vec{R})w(q^k_{i+1},q^k_{j})w(q^k_{i},q^k_{j+1}) \\
& \quad\quad + A_k(i,j)f_{0ij}(\vec{R})(1-w(q^k_{i+1},q^k_{j}))(1-w(q^k_{i},q^k_{j+1})) \\
& \quad\quad + A_k(i,j)f_{\iota ij}(\vec{R})w(q^k_{i+1},q^k_{j}))(1-w(q^k_{i},q^k_{j+1})) \\
& \quad\quad + A_k(i,j)f_{\iota ij}(\vec{R})(1-w(q^k_{i+1},q^k_{j}))w(q^k_{i},q^k_{j+1})
\end{align*}
\end{corollary}

\begin{proof}
If the pair $(i,j)$ is $1$-definite, then the RHS equals $A_k(i,j)f_{1ij}(\vec{R})$, which by 
Lemma~\ref{lem:nondiagbound} is indeed a lower bound on the integral. The other cases follow similarly.
\end{proof}

Note that the fourth term of the lower bound in Corollary \ref{cor:nondiagbound} is \emph{redundant} -- 
it will equal zero for any integral solution for $\vec{w}$ by monotonicity. Still, the term allows us to gain some more 
strength in the LP relaxation of the program, so we retain it in our final formulation.

Given a gauge $(\vec{q}^k)$, a lower bound function $f_{ij}(\vec{R},\vec{w})$ for each $1 \leq j \leq i \leq n$ is then provided 
by Lemma~\ref{lem:diagbound} and Corollary \ref{cor:nondiagbound}. To linearize the objective function, we again consider 
incorporating the relevant variables from the degree $3$ Sherali-Adams lift of the problem, with their defining inequalities. 

For the objective, we consider variables:
\begin{enumerate}
\item $w^2$ corresponding to terms of type $w(q_{i+1},q_j)\cdot w(q_i,q_{j+1})$,
\item $Rw$ corresponding to terms of type $R(q_\ell) \cdot w(q_{i+1},q_j)$ or $R(q_\ell) \cdot w(q_i,q_{j+1})$, and
\item $Rw^2$ corresponding to terms of type $R(q_\ell) \cdot w(q_{i+1},q_j) \cdot w(q_i,q_{j+1})$.
\end{enumerate}
For $w^2$, the Sherali-Adams inequalities are then:
\begin{align}
\forall 1 \leq j < i \leq n, \quad\quad\quad\quad & \label{con:w2Low} \\
-w(q_{i+1},q_j) + w^2(i,j) & \leq 0 \nonumber \\
-w(q_i,q_{j+1}) +w^2(i,j) & \leq 0 \nonumber \\
w(q_{i+1},q_j) + w(q_i,q_{j+1}) - w^2(i,j) & \leq 1 \nonumber
\end{align}
In turn, for $Rw$, the Sherali-Adams inequalities are given:
\begin{align}
\forall 1 \leq j < i \leq n, \forall \ell \in \{i,i+1,j,j+1\}, \forall (s,t) \in \{(&i+1,j),(i,j+1)\}, \label{con:RwLow} \\ 
-R(q_\ell) + Rw(\ell,s,t) & \leq 0 \nonumber \\
-w(q_s,q_t) + Rw(\ell,s,t) & \leq 0 \nonumber \\
R(q_\ell) + w(q_s,q_t) - Rw(\ell,s,t) & \leq 1 \nonumber 
\end{align}
Finally, we have the Sherali-Adams inequalities for $Rw^2$:
\begin{align}
\forall 1 \leq j < i \leq n, \forall \ell \in \{i,i+1,j,j+1\}, \quad\quad\quad\quad & \label{con:Rw2Low} \\ 
-w^2(i,j) + Rw^2(\ell,i,j) & \leq 0 \nonumber \\
-Rw(\ell,i+1,j) + Rw^2(\ell,i,j) & \leq 0 \nonumber \\
-Rw(\ell,i,j+1) + Rw^2(\ell,i,j) & \leq 0 \nonumber \\
-w(q_{i+1},q_j) + w^2(i,j) + Rw(\ell,i+1,j) - Rw^2(\ell,i,j) & \leq 0 \nonumber \\
-w(q_i,q_{j+1}) + w^2(i,j) + Rw(\ell,i,j+1) - Rw^2(\ell,i,j) & \leq 0 \nonumber \\
-R(q_\ell) + Rw(\ell,i+1,j) + Rw(\ell,i,j+1) - Rw^2(\ell,i,j) & \leq 0 \nonumber \\
R(q_\ell) + w(q_{i+1},q_j) + w(q_i,q_{j+1}) - Rw(\ell,i+1,j) - Rw(\ell,i,j+1) - w^2(i,j) + Rw^2(\ell,i,j) & \leq 1 \nonumber 
\end{align}

For the defining constraints for $w$, (\ref{con:wx}) and (\ref{con:wy}), we linearize terms of the 
form $R(q_\ell) \cdot w(\bar{q}_i,\bar{q}_j)$ to $Rw(\ell,i,j)$, coinciding with the previously defined $Rw$ term 
whenever necessary. These terms have defining inequalities:
\begin{align}
\forall 1 \leq j < i \leq n, \forall \ell \in \{i,j\}, \quad\quad \label{con:RwBis} \\ 
-R(q_\ell) + Rw(\ell,i,j) & \leq 0 \nonumber \\
-w(q_i,q_j) + Rw(\ell,i,j) & \leq 0 \nonumber \\
R(q_\ell) + w(q_i,q_j) - Rw(\ell,i,j) & \leq 1 \nonumber 
\end{align}

Finally, we again impose the monotonicity constraints for $w$, despite their redundancy. This implies that our lower 
bounding MILP has the following form:
\begin{align}
\min_{k \in \{1,2,...,N\}} \min_{\vec{R},\vec{w},w^2,Rw,Rw^2} & \quad \sum_{1 \leq i \leq n} A_k(i,i) \cdot f_{ii}(\vec{R},\vec{w}) + 2 \cdot \sum_{1 \leq j < i \leq n} A_k(i,j) \cdot f_{ij}(\vec{R},\vec{w})  \label{opt:lower} \\
\text{subject to} &  \quad (\ref{con:Rconcave}),(\ref{con:Ropt1}),(\ref{con:Ropt2}),(\ref{con:Ropt3}),(\ref{con:wx}),(\ref{con:wy}),(\ref{con:w2Low}),(\ref{con:RwLow}),(\ref{con:Rw2Low}),(\ref{con:RwBis}) \nonumber \\
& \quad w(\bar{q}_i,\bar{q}_j) \leq w(\bar{q}_{i+1},\bar{q}_j) \ \forall \ 1 \leq j < i < n \nonumber \\
& \quad w(\bar{q}_i,\bar{q}_j) \geq w(\bar{q}_i,\bar{q}_{j+1}) \ \forall \ 1 \leq j < i+1 \leq n \nonumber \\
& \quad w(q_i,q_i) = 0 \ \forall \ 1 \leq i \leq n \nonumber \\
& \quad \vec{R} \in [0,1]^{n+1} \nonumber \\ 
& \quad \vec{w} \in \{0,1\}^{\binom{n+1}{2}} \nonumber \\
& \quad w^2, Rw, Rw^2 \geq 0 \nonumber
\end{align}

\subsection{Technical Notes on Performance}\label{sec:perf}

For computing an upper bound by searching approximate worst-case distributions and for computing a lower bound, we 
solve a number of mixed-integer linear programs. As we either approximate or lower bound a Riemann integral, more refined 
gauges provide better guarantees. On the other hand, mixed-integer linear programming is of course NP-hard in general, and for a 
gauge with approximately $\geq 45$ intervals, we have $\geq 1000$ binary variables. Therefore, we must select our gauge 
while considering whether the family of MILPs is practically solvable.

For both programs (\ref{opt:approx2}) and (\ref{opt:lower}), we enforce inequalities corresponding to the monotonicity property of $w$ implied 
by Lemma~\ref{lem:wMon}. Even though such inequalities are redundant, we keep in mind that MILP solvers have to \emph{derive} 
valid inequalities, and in general may not be able to deduce families of inequalities implied by the specific properties of the problem. 
Indeed, the inclusion of monotonicity inequalities for $w$ results in considerable speed-up for the programs (\ref{opt:approx2}) and (\ref{opt:lower}).

Also, for a fixed number $n$ of intervals under consideration, (\ref{opt:approx2}) immediately provides a family of gauges. On the other 
hand, after selecting the number $N$ of optimal intervals, we still need to select gauges for (\ref{opt:lower}) for each $k \in \{1,2,...,N\}$. 
An immediate candidate is the ``approximately uniform'' gauge. For such a gauge, when $k = 1$, we divide $[1/N,1]$ into $n-1$ equal size 
intervals. Likewise, when $k = N$, we divide $[0,1-1/N]$ into $n-1$ equal size intervals. If instead $1 < k < N$, we choose $m$ such that:
$$ m \in \arg \min_{1 < \mu < n-1} \left| \frac{k-1}{N \cdot \mu} - \frac{N-k}{N \cdot (n-\mu-1)} \right| $$
We then divide $[0,(k-1)/N]$ into $m$ equal size intervals, and $[k/N,1]$ into $n-m-1$ equal size intervals. 

\begin{figure}[h]
\begin{centering}
\input{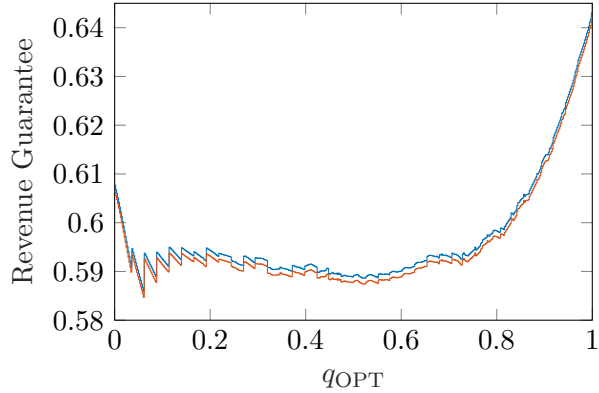}
\caption{Computation of (\ref{opt:lower}) with the approximately uniform gauge for $n = 40$ and $N = 500$. 
Values of solver output are shown in blue, while the lower bounds they imply after correcting for the tolerance of 
the solver are shown in orange.  
This computation implies a lower bound of $.5847$ via the lower bound for revenue when $q_\text{OPT} \in [.06,.062]$.}
\label{fig:n40}
\end{centering}
\end{figure}

While straightforward, this choice of gauge is problematic. To illustrate the issue, a computation of (\ref{opt:lower}) with the approximately uniform gauge 
for $n = 40$ and $N = 500$ at $99.8\%$ relative tolerance is shown in Figure \ref{fig:n40}. 
As seen in the figure, the approximately uniform gauge results in ``jagged'' behaviour for derived lower bounds when the revenue 
curve attains its maximum on (approximately) $[0,.2]$. The upwards kinks occur roughly when $k \rightarrow k+1$ causes $m \rightarrow m+1$. 
This implies that, for some initial segment of $[0,1]$, the quality of our lower bounds improve when we add more intervals in the segment $[0,(k-1)/N]$. 
As evidenced by Figure \ref{fig:n40}, somehow smoothing the jagged behaviour of the lower bound curve would allow us to improve 
our lower bound on the revenue of the ERM mechanism.

So we consider a modification of the approximately uniform gauge, \emph{square weighing} the gauge on $[0,1/2]$. In particular, 
for $k < N/2$, we instead choose $m$ such that
$$ m \in \arg \min_{1 < \mu < n-1} \left| \frac{k-1}{N \cdot \mu^2} - \frac{N-k}{N \cdot (n-\mu-1)^2} \right|.$$

Unfortunately, using this square-weighted gauge results in considerable slowdown of computations, when $k \lesssim N/10$. For 
this reason, we lower the relative efficiency guarantees of our solver when $k \leq N/10$. This results in a jump ``discontinuity'' in our 
computed revenue \emph{guarantees}, but this effect is not strong enough to overpower improvements on our final lower bound 
due to choice of gauge.

\begin{figure}[h]
\begin{centering}
\input{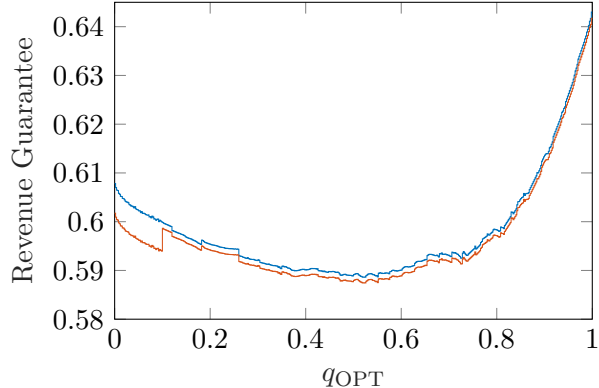}
\caption{Computation of (\ref{opt:lower}) with the square weighted gauge for $n = 40$ and $N = 500$, presented 
as in Figure \ref{fig:n40}. The square weighted gauge smoothes our lower bound estimates on $[0,1/2]$, providing a greater lower bound of $.5874$.}
\label{fig:gaugeDiff}
\end{centering}
\end{figure}

For a comparison of the quality of lower bounds provided by the approximately uniform gauge and the (partially) square-weighted 
gauge, we show in Figure \ref{fig:gaugeDiff} a computation of (\ref{opt:lower}) with the square-weighted gauge 
for $n = 40$ and $N = 500$, running our solver at $99.8\%$ relative tolerance for $k > 50$ and at $99\%$ relative tolerance 
for $k \leq 50$. As seen, the derived lower bounds are smoothed on the initial 
segment of $[0,1]$ by the weighing, and despite the jump in the lower bounds due to change in tolerance at $k = 50$, 
the quality of the lower bounds we obtain increase. The reason for why such a weighing 
works is unknown to us; indeed, we found the square-weighing rule by trial-and-error.

\section{Results: Lower and Upper Bounds}\label{sec:res}

We are now ready to present lower and upper bounds on the performance of the ERM mechanism with two samples.
We compute (\ref{opt:approx2}) and (\ref{opt:lower}) using MATLAB + CPLEX as our MILP solver of choice\footnote{Our code is available at \url{https://meteahunbay.com/files/code-twoSampleMILP.zip}.}. We compute~(\ref{opt:approx2}) for $n = 80$, obtaining an approximate \emph{conditional\footnote{On $\arg \max_{q\in [0,1]} R(q)$.} minimum expected revenue curve}. Each 
computation for $k \in \{1,...,81\}$ also provides us with an approximately minimal distribution; given primal solution $(R,w,Rw)$ 
to (\ref{opt:approx2}) for $k$, we consider the minimum concave function $R_k$ such that $R_k(j/n) = R(j/n)$ for any $j = \{0,1,...,80\}$. 
By numerically evaluating the integral (\ref{eqn:qSpace}) in Mathematica for each such $R_k$, we obtain upper bounds on the 
performance of the ERM mechanism.

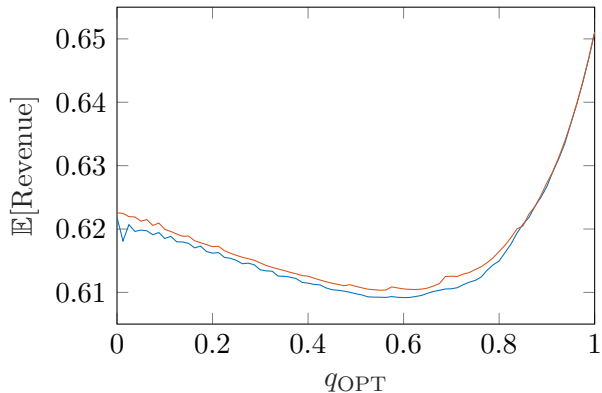
\begin{figure}[H]
\begin{centering}
%
%
\definecolor{mycolor1}{rgb}{0.00000,0.44700,0.74100}%
\definecolor{mycolor2}{rgb}{0.85000,0.32500,0.09800}%
\begin{tikzpicture}

\begin{axis}[%
width=2.5in,
height=1.66in,
at={(0.758in,0.481in)},
scale only axis,
xmin=0,
xmax=1,
xlabel style={font=\color{white!15!black}},
xlabel={$q_{\text{OPT}}$},
ymin=0.605,
ymax=0.655,
ytick = {.61,.62,.63,.64,.65},
ylabel style={font=\color{white!15!black}},
ylabel={$\mathbb{E}[\text{Revenue}]$},
axis background/.style={fill=white}
]
\addplot [color=mycolor1]
  table[row sep=crcr]{%
0	0.621847597996752\\
0.0125	0.618128909867634\\
0.025	0.620690102817644\\
0.0375	0.619613958270099\\
0.05	0.619831820775858\\
0.0625	0.619718700140898\\
0.075	0.619103054103899\\
0.0875	0.619445319983182\\
0.1	0.618501626232825\\
0.1125	0.618841183316504\\
0.125	0.617990542038804\\
0.1375	0.617954072351711\\
0.15	0.617732631840794\\
0.1625	0.617021393380992\\
0.175	0.617308693138422\\
0.1875	0.616437767741248\\
0.2	0.616191381094405\\
0.2125	0.616284152431608\\
0.225	0.615521093124485\\
0.2375	0.615376032196026\\
0.25	0.615049930120273\\
0.2625	0.614521765820342\\
0.275	0.614615267318794\\
0.2875	0.614354198597016\\
0.3	0.613558917073344\\
0.3125	0.613399844205201\\
0.325	0.613337215018357\\
0.3375	0.612562196327838\\
0.35	0.612547205613761\\
0.3625	0.612404096884759\\
0.375	0.612180119055477\\
0.3875	0.611585699917948\\
0.4	0.611458562700086\\
0.4125	0.611266346354025\\
0.425	0.611185227168251\\
0.4375	0.610715446837248\\
0.45	0.610423574782593\\
0.4625	0.610342104189977\\
0.475	0.610206047586022\\
0.4875	0.610024335049522\\
0.5	0.609773264821277\\
0.5125	0.609609402181806\\
0.525	0.60930638201951\\
0.5375	0.60925098816344\\
0.55	0.609244252833807\\
0.5625	0.609205072708021\\
0.575	0.609344678031428\\
0.5875	0.60922474676779\\
0.6	0.609180986376231\\
0.6125	0.609198184287442\\
0.625	0.609349280149051\\
0.6375	0.609518739664064\\
0.65	0.609856098526185\\
0.6625	0.610141417031262\\
0.675	0.610327009704171\\
0.6875	0.610541002743962\\
0.7	0.610570092170196\\
0.7125	0.610741159480187\\
0.725	0.611220290826087\\
0.7375	0.611610550568237\\
0.75	0.611897848655076\\
0.7625	0.612427984294983\\
0.775	0.613506911398849\\
0.7875	0.614345647836997\\
0.8	0.614959146687186\\
0.8125	0.616295375224263\\
0.825	0.617644545631548\\
0.8375	0.619371891847199\\
0.85	0.620761396332245\\
0.8625	0.621820652173912\\
0.875	0.623594282834101\\
0.8875	0.624997065727696\\
0.9	0.626702586206893\\
0.9125	0.629047303082188\\
0.925	0.631107826576574\\
0.9375	0.633470833333333\\
0.95	0.636612331610276\\
0.9625	0.639729437229436\\
0.975	0.643088942307693\\
0.9875	0.646843354430375\\
1	0.6510234375\\
};

\addplot [color=mycolor2]
  table[row sep=crcr]{%
0	0.622555852974572\\
0.0125	0.622476777780105\\
0.025	0.621949407129057\\
0.0375	0.621905021778558\\
0.05	0.621241575011267\\
0.0625	0.621482143961759\\
0.075	0.620560170720463\\
0.0875	0.620945626035401\\
0.1	0.61995313503559\\
0.1125	0.619620192733837\\
0.125	0.619211771435211\\
0.1375	0.618860459478917\\
0.15	0.618888811244881\\
0.1625	0.618144476264573\\
0.175	0.617845860648914\\
0.1875	0.617531913612085\\
0.2	0.617216237467852\\
0.2125	0.617284569524958\\
0.225	0.616575771579944\\
0.2375	0.616181239697985\\
0.25	0.615854323208065\\
0.2625	0.615539635080567\\
0.275	0.615313150454835\\
0.2875	0.615069088519002\\
0.3	0.614647265916255\\
0.3125	0.614243636227304\\
0.325	0.613946736709944\\
0.3375	0.613694275293204\\
0.35	0.613457478770962\\
0.3625	0.613168296794351\\
0.375	0.612938280692933\\
0.3875	0.612662930949663\\
0.4	0.612555534043416\\
0.4125	0.612255580706186\\
0.425	0.61195379176307\\
0.4375	0.611688112522375\\
0.45	0.611446868041451\\
0.4625	0.611259328021871\\
0.475	0.611063478507371\\
0.4875	0.611217049372912\\
0.5	0.610955121823257\\
0.5125	0.610730310489126\\
0.525	0.610507967355456\\
0.5375	0.610429503700358\\
0.55	0.610349854453496\\
0.5625	0.610373124596528\\
0.575	0.610862820772694\\
0.5875	0.610700059872488\\
0.6	0.610546480134837\\
0.6125	0.610482563923372\\
0.625	0.610456715719923\\
0.6375	0.610519931353598\\
0.65	0.610681311710439\\
0.6625	0.611013161557775\\
0.675	0.611384077729827\\
0.6875	0.612513007837986\\
0.7	0.612552465603761\\
0.7125	0.612510391046819\\
0.725	0.612886009595031\\
0.7375	0.613133217992967\\
0.75	0.61361436292613\\
0.7625	0.614048733994564\\
0.775	0.614702451816799\\
0.7875	0.615466521662223\\
0.8	0.61646041984118\\
0.8125	0.617454878858379\\
0.825	0.6186636084069\\
0.8375	0.620036404912223\\
0.85	0.620487995063486\\
0.8625	0.622359465102598\\
0.875	0.623551292373271\\
0.8875	0.625302672270561\\
0.9	0.627285476943963\\
0.9125	0.62905937331027\\
0.925	0.631340467068702\\
0.9375	0.633847287730874\\
0.95	0.636667650290906\\
0.9625	0.63972047472825\\
0.975	0.643163389556013\\
0.9875	0.646928493459534\\
1	0.651103485328401\\
};

\end{axis}
\end{tikzpicture}%
\caption{Results of computation of (\ref{opt:approx2}) for $n = 80$. The values of (\ref{opt:approx2}) conditional on $q_{OPT}$ is shown in 
blue, while upper bounds on the performance of the ERM mechanism obtained from primal solutions are shown in orange.}
\label{fig:n80UB}
\end{centering}
\end{figure}

The results of this computation is shown in Figure \ref{fig:n80UB}. Numerically computing the integral 
\ref{eqn:qSpace} for each primal solution we obtain, the internal error estimates provided by Mathematica are $\leq 10^{-6}$ for 
each integral approximation. Via these computations, we find that our primal solution for $n = 80$ 
and $q_\text{OPT} = 44/80$ provides a regular revenue curve for which the ERM mechanism obtains  $\leq .61035$ times the 
optimal revenue. Furthermore, we are able to inspect the form of minimal distributions themselves. Figure \ref{fig:n80ex} shows the form 
of such distributions -- it appears they are closely approximated by piecewise linear functions on at most three intervals ($3$-piecewise 
linear functions).

\begin{figure}[H]
\centering
\begin{minipage}{.49\textwidth}
  \centering
%
%
\definecolor{mycolor1}{rgb}{0.00000,0.44700,0.74100}%
\begin{tikzpicture}

\begin{axis}[%
width=2in,
height=1.33in,
at={(0.594in,0.28in)},
scale only axis,
xmin=0,
xmax=1,
xlabel style={font=\color{white!15!black}},
xlabel={$q$},
ymin=0,
ymax=1,
ylabel style={font=\color{white!15!black}},
ylabel={$R(q)$},
axis background/.style={fill=white}
]
\addplot [color=mycolor1]
  table[row sep=crcr]{%
0	0\\
0.0125	0.0714285714285713\\
0.025	0.142857142857143\\
0.0375	0.214285714285714\\
0.05	0.285714285714286\\
0.0625	0.357142857142857\\
0.075	0.428571428571429\\
0.0875	0.5\\
0.1	0.571428571428572\\
0.1125	0.642857142857143\\
0.125	0.714285714285715\\
0.1375	0.785714285714286\\
0.15	0.857142857142857\\
0.1625	0.928571428571429\\
0.175	1\\
0.1875	0.986916172219686\\
0.2	0.973832344439372\\
0.2125	0.960748516659059\\
0.225	0.947664688878745\\
0.2375	0.934580861098431\\
0.25	0.921497033318117\\
0.2625	0.908413205537804\\
0.275	0.89532937775749\\
0.2875	0.882245549977176\\
0.3	0.869161722196862\\
0.3125	0.856077894416549\\
0.325	0.842994066636235\\
0.3375	0.829910238855921\\
0.35	0.816826411075607\\
0.3625	0.803742583295294\\
0.375	0.79065875551498\\
0.3875	0.777574927734666\\
0.4	0.764491099954352\\
0.4125	0.751407272174038\\
0.425	0.738323444393725\\
0.4375	0.725239616613411\\
0.45	0.712155788833098\\
0.4625	0.699071961052784\\
0.475	0.685988133272471\\
0.4875	0.672904305492158\\
0.5	0.659820477711844\\
0.5125	0.646736649931531\\
0.525	0.633652822151217\\
0.5375	0.620568994370904\\
0.55	0.60748516659059\\
0.5625	0.594401338810277\\
0.575	0.581317511029964\\
0.5875	0.568233683249651\\
0.6	0.555149855469337\\
0.6125	0.542066027689024\\
0.625	0.52898219990871\\
0.6375	0.515898372128397\\
0.65	0.502814544348084\\
0.6625	0.489730716567771\\
0.675	0.476646888787457\\
0.6875	0.463563061007144\\
0.7	0.450479233226831\\
0.7125	0.437395405446518\\
0.725	0.424311577666204\\
0.7375	0.411227749885891\\
0.75	0.398143922105578\\
0.7625	0.385060094325265\\
0.775	0.371976266544952\\
0.7875	0.358892438764638\\
0.8	0.345808610984325\\
0.8125	0.332724783204012\\
0.825	0.319640955423699\\
0.8375	0.306557127643386\\
0.85	0.293473299863072\\
0.8625	0.280389472082759\\
0.875	0.267305644302446\\
0.8875	0.254221816522133\\
0.9	0.241137988741819\\
0.9125	0.228054160961506\\
0.925	0.214970333181193\\
0.9375	0.20063056151697\\
0.95	0.167674637433366\\
0.9625	0.131289444278315\\
0.975	0.0936642586967572\\
0.9875	0.0560390731151989\\
1	0.00651988017857788\\
};

\end{axis}
\end{tikzpicture}%
\end{minipage}
\begin{minipage}{.49\textwidth}
  \centering
%
%
\definecolor{mycolor1}{rgb}{0.00000,0.44700,0.74100}%
\begin{tikzpicture}

\begin{axis}[%
width=2in,
height=1.33in,
at={(0.65in,0.343in)},
scale only axis,
xmin=0,
xmax=1,
xlabel style={font=\color{white!15!black}},
xlabel={$q$},
ymin=0,
ymax=1,
ylabel style={font=\color{white!15!black}},
ylabel={$R(q)$},
axis background/.style={fill=white}
]
\addplot [color=mycolor1]
  table[row sep=crcr]{%
0	0\\
0.0125	0.0227272727272729\\
0.025	0.0454545454545454\\
0.0375	0.0681818181818179\\
0.05	0.0909090909090904\\
0.0625	0.113636363636363\\
0.075	0.136363636363635\\
0.0875	0.159090909090908\\
0.1	0.18181818181818\\
0.1125	0.204545454545453\\
0.125	0.227272727272726\\
0.1375	0.249999999999998\\
0.15	0.272727272727271\\
0.1625	0.295454545454543\\
0.175	0.318181818181816\\
0.1875	0.340909090909088\\
0.2	0.363636363636361\\
0.2125	0.386363636363633\\
0.225	0.409090909090906\\
0.2375	0.431818181818178\\
0.25	0.454545454545451\\
0.2625	0.477272727272724\\
0.275	0.499999999999996\\
0.2875	0.522727272727269\\
0.3	0.545454545454542\\
0.3125	0.568181818181815\\
0.325	0.590909090909087\\
0.3375	0.61363636363636\\
0.35	0.636363636363633\\
0.3625	0.659090909090906\\
0.375	0.681818181818179\\
0.3875	0.704545454545451\\
0.4	0.727272727272724\\
0.4125	0.749999999999997\\
0.425	0.77272727272727\\
0.4375	0.795454545454543\\
0.45	0.818181818181816\\
0.4625	0.840909090909089\\
0.475	0.863636363636362\\
0.4875	0.886363636363635\\
0.5	0.909090909090908\\
0.5125	0.931818181818181\\
0.525	0.954545454545454\\
0.5375	0.977272727272727\\
0.55	1\\
0.5625	0.974594831362242\\
0.575	0.949189662724485\\
0.5875	0.923784494086727\\
0.6	0.898379325448969\\
0.6125	0.872974156811212\\
0.625	0.847568988173454\\
0.6375	0.822163819535697\\
0.65	0.796758650897939\\
0.6625	0.771353482260182\\
0.675	0.745948313622424\\
0.6875	0.720543144984667\\
0.7	0.69513797634691\\
0.7125	0.669732807709152\\
0.725	0.644327639071395\\
0.7375	0.618922470433638\\
0.75	0.59351730179588\\
0.7625	0.568112133158123\\
0.775	0.542706964520366\\
0.7875	0.517301795882609\\
0.8	0.491896627244851\\
0.8125	0.466491458607094\\
0.825	0.441086289969337\\
0.8375	0.41568112133158\\
0.85	0.390275952693823\\
0.8625	0.364870784056065\\
0.875	0.339465615418308\\
0.8875	0.314060446780551\\
0.9	0.288655278142794\\
0.9125	0.263250109505036\\
0.925	0.237844940867279\\
0.9375	0.212439772229522\\
0.95	0.187034603591765\\
0.9625	0.161629434954008\\
0.975	0.13622426631625\\
0.9875	0.0764989307155183\\
1	0.000592615495607213\\
};

\end{axis}
\end{tikzpicture}%
\end{minipage}
\caption{Revenue curves of approximately minimal distributions (conditional on $q_{OPT}$) obtained 
from (\ref{opt:approx2}) for $n = 80$ and $k = 15,45$.}
\label{fig:n80ex}
\end{figure}
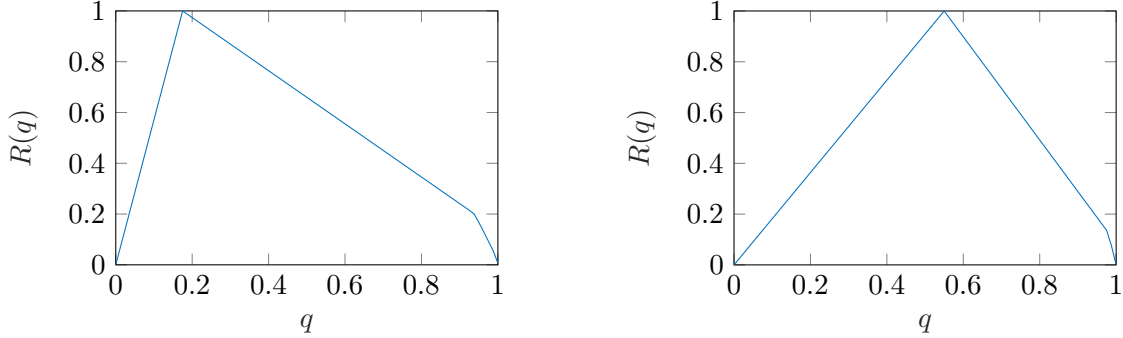

This motivates us to inspect the efficiency of the ERM mechanism for piecewise linear functions on $\leq$ three intervals. 
Note that any such function $R$, normalised such that $\max_{q \in [0,1]} R(q) = 1$, is the minimum concave function 
that contains $(0,0),(q_{OPT},1),(q_2,r_2),(1,r_3)$ in its hypograph for some $q_{OPT},q_2,r_2,r_3 \in [0,1]$. Given 
such $q_{OPT},q_2,r_2,r_3$, we denote by $R(\cdot|q_{OPT},q_2,r_2,r_3)$ the minimum concave function whose 
hypograph contains $(0,0),(q_{OPT},1),(q_2,r_2),(1,r_3)$.

Given $q_{OPT}$, the minimum performance of the ERM mechanism for functions of the form $R(\cdot|q_{OPT},q_2,r_2,r_3)$ is given by:
\begin{equation}\label{eqn:eta}
\eta(q_{OPT}) = \min_{q_2,r_2,r_3} \int_{(x,y) \in [0,1]^2} \phi_{R(\cdot|q_{OPT},q_2,r_2,r_3)}(x,y) \cdot d(x,y)
\end{equation}

For $q_{OPT} \in \{0,1/80,2/80,...,1\}$ we attempt to approximate $\eta(q_{OPT})$ by grid search, evaluating the 
minimum in (\ref{eqn:eta}) for
\begin{align*}
q_2 & = k/80 &  71 \leq & \ k \leq 80\\
r_2 & = k/1000 &   0 \leq & \ k \leq 300\\
r_3 & = k/1000 &  0 \leq & \ k \leq 18, k \in \mathbb{Z}.
\end{align*}
Note that we do not grid the entire cube $[0,1]^3$. We instead choose to restrict the bound of our grid search by the 
form of our primal solutions for (\ref{opt:approx2}). This is because we have found that the optimization problem 
for $\eta(q_{OPT})$ is very poorly behaved. In particular, solutions obtained by searching over a coarse grid and then 
improving via gradient descent have provided greater (by $\gtrsim .01$) upper bounds than simply searching on a fine grid. 
In addition, to keep computation costs low we have opted to restrict our search space.

\begin{figure}[H]
\begin{centering}
%
%
\definecolor{mycolor1}{rgb}{0.00000,0.44700,0.74100}%
\definecolor{mycolor2}{rgb}{0.85000,0.32500,0.09800}%
\begin{tikzpicture}

\begin{axis}[%
width=2.5in,
height=1.66in,
at={(0.758in,0.481in)},
scale only axis,
xmin=0,
xmax=1,
xlabel style={font=\color{white!15!black}},
xlabel={$q_{\text{OPT}}$},
ymin=0.605,
ymax=0.655,
ytick = {.61,.62,.63,.64,.65},
ylabel style={font=\color{white!15!black}},
ylabel={$\mathbb{E}[\text{Revenue}]$},
axis background/.style={fill=white}
]
\addplot [color=mycolor1]
  table[row sep=crcr]{%
0	0.622555852974572\\
0.0125	0.622476777780105\\
0.025	0.621949407129057\\
0.0375	0.621905021778558\\
0.05	0.621241575011267\\
0.0625	0.621482143961759\\
0.075	0.620560170720463\\
0.0875	0.620945626035401\\
0.1	0.61995313503559\\
0.1125	0.619620192733837\\
0.125	0.619211771435211\\
0.1375	0.618860459478917\\
0.15	0.618888811244881\\
0.1625	0.618144476264573\\
0.175	0.617845860648914\\
0.1875	0.617531913612085\\
0.2	0.617216237467852\\
0.2125	0.617284569524958\\
0.225	0.616575771579944\\
0.2375	0.616181239697985\\
0.25	0.615854323208065\\
0.2625	0.615539635080567\\
0.275	0.615313150454835\\
0.2875	0.615069088519002\\
0.3	0.614647265916255\\
0.3125	0.614243636227304\\
0.325	0.613946736709944\\
0.3375	0.613694275293204\\
0.35	0.613457478770962\\
0.3625	0.613168296794351\\
0.375	0.612938280692933\\
0.3875	0.612662930949663\\
0.4	0.612555534043416\\
0.4125	0.612255580706186\\
0.425	0.61195379176307\\
0.4375	0.611688112522375\\
0.45	0.611446868041451\\
0.4625	0.611259328021871\\
0.475	0.611063478507371\\
0.4875	0.611217049372912\\
0.5	0.610955121823257\\
0.5125	0.610730310489126\\
0.525	0.610507967355456\\
0.5375	0.610429503700358\\
0.55	0.610349854453496\\
0.5625	0.610373124596528\\
0.575	0.610862820772694\\
0.5875	0.610700059872488\\
0.6	0.610546480134837\\
0.6125	0.610482563923372\\
0.625	0.610456715719923\\
0.6375	0.610519931353598\\
0.65	0.610681311710439\\
0.6625	0.611013161557775\\
0.675	0.611384077729827\\
0.6875	0.612513007837986\\
0.7	0.612552465603761\\
0.7125	0.612510391046819\\
0.725	0.612886009595031\\
0.7375	0.613133217992967\\
0.75	0.61361436292613\\
0.7625	0.614048733994564\\
0.775	0.614702451816799\\
0.7875	0.615466521662223\\
0.8	0.61646041984118\\
0.8125	0.617454878858379\\
0.825	0.6186636084069\\
0.8375	0.620036404912223\\
0.85	0.620487995063486\\
0.8625	0.622359465102598\\
0.875	0.623551292373271\\
0.8875	0.625302672270561\\
0.9	0.627285476943963\\
0.9125	0.62905937331027\\
0.925	0.631340467068702\\
0.9375	0.633847287730874\\
0.95	0.636667650290906\\
0.9625	0.63972047472825\\
0.975	0.643163389556013\\
0.9875	0.646928493459534\\
1	0.651103485328401\\
};

\addplot [color=mycolor2]
  table[row sep=crcr]{%
0	0.623034543737331\\
0.0125	0.623404027951976\\
0.025	0.622990091573464\\
0.0375	0.622592053095435\\
0.05	0.622265499188368\\
0.0625	0.62183106278011\\
0.075	0.621414221907146\\
0.0875	0.621087193671902\\
0.1	0.620726935462941\\
0.1125	0.620404850459576\\
0.125	0.62005023290497\\
0.1375	0.619640143352783\\
0.15	0.619226968759307\\
0.1625	0.618835076828988\\
0.175	0.618439710124442\\
0.1875	0.618070842501744\\
0.2	0.617707506708158\\
0.2125	0.617383478087663\\
0.225	0.617005028246991\\
0.2375	0.616670031887117\\
0.25	0.61636579174184\\
0.2625	0.616057143230771\\
0.275	0.615696652794261\\
0.2875	0.615329923241252\\
0.3	0.614960924298566\\
0.3125	0.614625124040463\\
0.325	0.614288622981592\\
0.3375	0.613997075258352\\
0.35	0.613664658031047\\
0.3625	0.613385766195224\\
0.375	0.613129545422165\\
0.3875	0.61295642022825\\
0.4	0.612653909445719\\
0.4125	0.612452594552127\\
0.425	0.612168678566319\\
0.4375	0.611895738071737\\
0.45	0.611640115236311\\
0.4625	0.611409607438679\\
0.475	0.611200111051094\\
0.4875	0.611005510421505\\
0.5	0.610854254378486\\
0.5125	0.610715792117957\\
0.525	0.610622539107796\\
0.5375	0.610576497887965\\
0.55	0.610543513548706\\
0.5625	0.61055689107628\\
0.575	0.61046147651954\\
0.5875	0.610406961117867\\
0.6	0.61038160948084\\
0.6125	0.610386480506424\\
0.625	0.610427659908124\\
0.6375	0.610505835395357\\
0.65	0.61063233000858\\
0.6625	0.610802218028318\\
0.675	0.611036334901654\\
0.6875	0.611328372706209\\
0.7	0.611702155020256\\
0.7125	0.612131209014087\\
0.725	0.612603409595528\\
0.7375	0.613104197023358\\
0.75	0.613745898403739\\
0.7625	0.614365277110541\\
0.775	0.614893537108126\\
0.7875	0.615521420900841\\
0.8	0.616268884254195\\
0.8125	0.617129064082967\\
0.825	0.618123701940794\\
0.8375	0.619234097347865\\
0.85	0.620475148707332\\
0.8625	0.6218591150259\\
0.875	0.623490655630533\\
0.8875	0.62512961064511\\
0.9	0.626980642617837\\
0.9125	0.629032118949498\\
0.925	0.631320840695515\\
0.9375	0.633846584339575\\
0.95	0.636647253340008\\
0.9625	0.639720474728256\\
0.975	0.643163389557241\\
0.9875	0.646928493459549\\
1	0.651103485328408\\
};

\end{axis}
\end{tikzpicture}%
\caption{Comparison of derived upper bounds conditional on $q_{OPT}$. The expected revenue of the ERM mechanism 
for distributions obtained from primal solutions of (\ref{opt:approx2}) are shown in blue, while the expected revenue of the 
ERM mechanism for minimal distributions obtained from a grid-search approximation of $\eta(q_{OPT})$ are shown in orange.} 
\label{fig:comp}
\end{centering}
\end{figure}
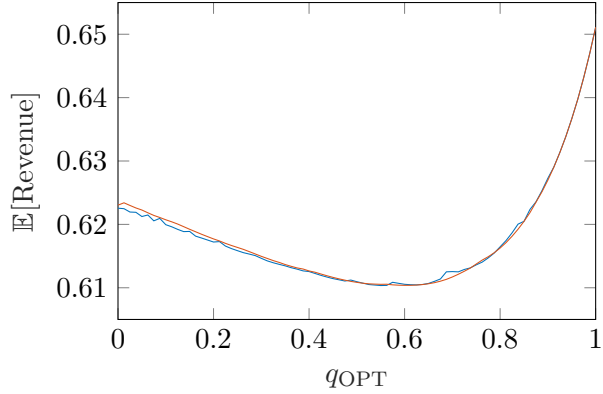

The results of our grid search are shown in Figure \ref{fig:comp}, juxtaposed with the upper bounds we obtain via 
solutions of (\ref{opt:approx2}). It appears that both curves approximate some convex shape on $[0,1]$. Since we do not 
have any approximation guarantees from our grid search, we are unable to tell if $\eta(\cdot)$ can or cannot provide stronger 
lower bounds than those derived via our primal solutions to (\ref{opt:approx2}) for any $q_{OPT}$ in general, but values we 
obtain via grid search do not improve on our upper bound of $.61035$.

\begin{figure}[H]
\begin{centering}
\input{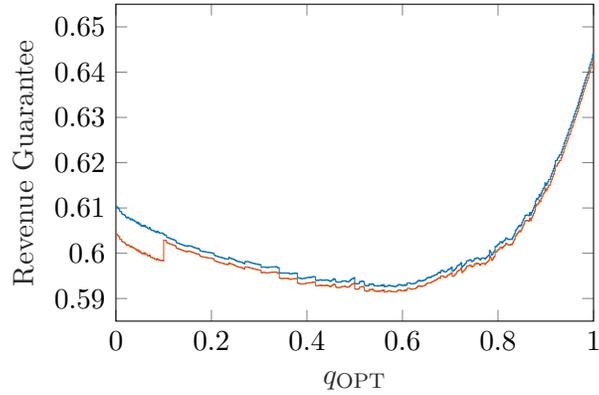}
\caption{Results of computation of (\ref{opt:lower}) for $n = 50$ and $N = 500$. The blue line shows the value of the 
objective for primal solutions found, while the orange line shows conditional lower bounds on the expected revenue of the 
auction. These results are corrected for tolerance of the MILP solver.}
\label{fig:n50LB}
\end{centering}
\end{figure}

Finally, we compute (\ref{opt:lower}) for $n = 50$ and $N = 500$, running our solver at $99.8\%$ relative tolerance for $k > 50$ 
and $99\%$ relative tolerance for $k \leq 50$.  The results of the computation are shown in Figure \ref{fig:n50LB}. Our results 
show that the ERM mechanism guarantees an expected revenue $\geq .5914$ times the optimal revenue.

\section{Conclusion}\label{sec:conc}

In this paper, we presented an MILP formulation to inspect the expected revenue of the ERM mechanism in the single item, 
single buyer, two sample setting. Working within this formulation has allowed us to greatly improve upon the known upper and 
lower bounds of the expected revenue guarantees of the ERM mechanism with two samples, and provided us with insights on 
what minimum revenue distributions may look like.

Despite the sheer number of binary variables involved, computations to certify our bounds were relatively \emph{cheap} -- on 
a ASUS ROG Zephyrus M (GU502GV) laptop, (\ref{opt:approx2}) for $n = 80$ took approximately a day to compute, while the 
computations to solve (\ref{opt:lower}) for $n = 50$ and $N = 500$ took around twelve days. Still, the exponential nature of the 
problem had become noticable around the values of $(n,N)$ we used. Therefore, we do not expect (\ref{opt:approx2}) and (\ref{opt:lower}) 
to be feasibly solvable for significantly finer gauges, disallowing major improvements on the bounds we have provided by simply 
solving (\ref{opt:approx2}) and (\ref{opt:lower}) for larger $n,N$.

That being said, it may be still possible to extract even stronger lower bounds within our framework. Lower bounds we may 
derive from solutions of (\ref{opt:approx2}) currently depend on the proof of Theorem \ref{thm:easyApprox}. For fixed $n$, 
our estimation of how much the value of (\ref{opt:approx2}) overestimates $\alpha$ is $2/(n-1) + (5n-6)/n^2$. For $n = 80$, this error 
estimate is $\lesssim .0869$, which means that our computations for (\ref{opt:approx2}) can only certify a lower bound of $.5210$. 
However, Figure \ref{fig:n80UB} and Figure \ref{fig:comp} hint that the actual error might in fact be much smaller than our 
estimate. Improving this estimate could then help certify stronger lower bounds on the revenue guarantees of the ERM 
mechanism with two samples.

There is also the question of what the minimal revenue curves for the ERM mechanism with two samples, conditional 
on $q_{OPT}$, actually look like. Figure \ref{fig:n80ex} and Figure \ref{fig:comp} suggest that $3$-piecewise linear functions 
may be \emph{close} to minimality, but we are currently unable to discern if they actually provide minimal instances. The 
separation of the two lines on $[0,.2]$ in Figure \ref{fig:comp} hint that the answer is \emph{no}, but the lines might have diverged 
simply because our search grid for $3$-piecewise linear functions was not sufficiently fine. Even if this is the case, however, it might 
be that the \emph{minimum} revenue curve is $3$-piecewise linear. 

Finally, we note that our formulation should extend naturally to the setting with $\geq 3$ samples. However, in such an extension, 
the number of binary variables would blow up exponentially as the number of samples increases for fixed number of intervals, $n$. 
This implies that the extension of (\ref{opt:approx2}) and (\ref{opt:lower}) to a setting with $\geq 3$ samples might not be feasibly 
solvable. Still, for settings in which the performance of solvers do not depreciate too much, our techniques should be readily applicable.

\bibliographystyle{acm}
\bibliography{bibliography}

\end{document}